\documentclass[sigconf,nonacm]{acmart}
\usepackage{soul,balance}
\usepackage{booktabs} %,amsthm,amsmath,amssymb} %,times} 	% to make equations look good
\usepackage{comment}
\usepackage{hyperref} % Allows for clickable references
\usepackage{graphicx} 	% Allows you to import images
\usepackage{float,wrapfig} 		    % Allows for control of float positions
\usepackage{mathtools,verbatim}      % For floor and ceil functions
\usepackage{caption}
\usepackage{subcaption}
\usepackage{color}
\usepackage{colortbl}
\usepackage{multicol}
\usepackage{xcolor}  % For colored cells in tables
\definecolor{lightgreen}{rgb}{0.56, 0.93, 0.56}  % Define light green
\definecolor{lightblue}{rgb}{0.68, 0.85, 0.90}
\usepackage{enumerate}
\usepackage{enumitem}
\usepackage[linesnumbered,ruled,vlined]{algorithm2e} 
\usepackage{multicol} % Required for the side-by-side layout

\usepackage[ruled, vlined,linesnumbered]{algorithm2e}

\usepackage{pifont}
\newcommand{\cmark}{\ding{51}}
\newcommand{\xmark}{\ding{55}}
\usepackage{multicol}
\usepackage{rotating}
\usepackage{graphicx}

\begin{document}

\title{%Scalable Fault-Tolerant Causal Delivery of Unicast Messages %icdcn
Byzantine-Tolerant Causal Unicast with Constant Message Space Overhead
%Scalable Byzantine-Tolerant Causal Delivery in Peer-to-Peer Systems
%Constant Message-Space Overhead for Byzantine-Tolerant Causal Delivery
}
%\subtitle{(Regular Paper)}
%\title{\LARGE \bf } % }

%\author{Anonymous Authors}
%\begin{comment}
\author{Purv Patel}
\affiliation{%
  \institution{University of Illinois Chicago}
  %\streetaddress{1 Th{\o}rv{\"a}ld Circle}
  \city{Chicago}
  \country{USA}
  }
\email{purvp2@uic.edu}

\author{Ajay D. Kshemkalyani}
\affiliation{%
  \institution{University of Illinois Chicago}
  %\streetaddress{1 Th{\o}rv{\"a}ld Circle}
  \city{Chicago}
  \country{USA}
  }
\email{ajay@uic.edu}
%\end{comment}

%\author{Anonymous authors}

%\ccsdesc{Theory of computation~Distributed algorithms}
%\ccsdesc[500]{Theory of computation~Distributed algorithms}
%\ccsdesc[500]{Theory of computation~Concurrent algorithms}
%\ccsdesc[300]{Computer systems organization~Dependable and fault-tolerant systems and networks}
%\keywords{Causal ordering, Unicasts, Byzantine fault-tolerance}

\begin{abstract}
Causal message ordering provides essential semantics for distributed applications, yet ensuring it within an asynchronous system subject to Byzantine failures presents fundamental theoretical and practical challenges. Prior research establishes that algorithms cannot guarantee both strong safety and liveness without using cryptography under these conditions. %without synchrony assumptions, while 
Existing Byzantine-tolerant solutions make synchrony assumptions or suffer from $O(n)$ message space overheads and $O(n^2)$ message complexity, where $n$ is the number of processes in the system, or use cryptography, but may not guarantee strong safety. In this paper, we present a novel Byzantine-tolerant causal ordering algorithm that achieves an optimal $O(1)$ application message space overhead for point-to-point messages. 
%Our approach builds upon the Sender Permission to Send (SPS) invariant of Almeida  \cite{almeida}, but fundamentally redesigns the memory management architecture into an \textit{Isolated-Buffer Optimistic Model}. To circumvent Byzantine state-pinning and head-of-line blocking attacks, we replace unified causal buffers with isolated, per-peer queues equipped with event-driven \textit{Cascading Space Evictions}. 
Our approach uses a Sender Permission to Send (SPS) invariant and an \textit{Isolated-Buffer Optimistic Model} for the memory management architecture. 
To circumvent Byzantine state-pinning and head-of-line blocking attacks, instead of unified causal buffers, we use isolated, per-peer queues equipped with event-driven \textit{Cascading Space Evictions}. 
%Our approach leverages an isolated-buffer optimistic architecture to decouple local memory management from global consensus, utilizing cascading space evictions to prevent Byzantine head-of-line blocking and state-pinning attacks. 
%To maintain causal integrity without appending linear histories, we separate communication into two planes: a lightweight data plane for application unicasts, and a control plane powered by an Asynchronous Byzantine Atomic Broadcast (ABAB) primitive. %The ABAB channel synchronizes a deterministic cryptographic hash ledger, allowing nodes to attach a constant-sized state proof to outgoing messages. 
% We formally prove that our algorithm guarantees liveness %and cryptographic integrity 
% by restricting causal safety violations exclusively to information laundering and premature optimistic evictions. This deliberate trade-off successfully balances causal ordering with high throughput, constant message space overhead, low computational overhead, and strictly bounded local space.
We formally prove that our algorithm guarantees system-wide liveness and satisfies a weakened causal safety abstraction called \textit{Congestion-Relaxed Causal Delivery}. Under this model, weak causal safety is strictly preserved for all honest-to-honest communications unless extreme network latency or Byzantine withholding attacks exceed quantifiable local buffer capacities, forcing optimistic queue evictions. This formal guarantee successfully balances causal ordering with high throughput, constant message space overhead, low computational overhead, and strictly bounded local space.
\end{abstract}


\maketitle

\section{Introduction}

The demand for robust, highly available distributed systems has grown drastically with the rise of cloud computing, distributed data stores, and collaborative interactive applications. In these environments, establishing a consistent order of events is critical to maintaining system semantics \cite{birman1987reliable}. Causal message ordering ensures that the delivery of messages respects the \textit{happens-before} relationship \cite{LLclock}, allowing systems to provide strong consistency guarantees without the overhead of total synchronization.

%Guaranteeing causal delivery in an asynchronous system subject to Byzantine failures remains a challenge. In a Byzantine fault model, compromised processes can act arbitrarily, dropping messages, forging causal histories, or executing resource-exhaustion attacks by withholding required control messages. Recent theoretical results have proven that achieving strict causal ordering (strong safety) alongside guaranteed message delivery (liveness) is mathematically impossible in a deterministic, asynchronous Byzantine environment \cite{DBLP:journals/tpds/MisraK24}. Consequently, distributed systems must trade strong safety for liveness, or rely on synchronous assumptions such as bounded latency.

%Furthermore, existing causal ordering algorithms typically suffer from high message space overheads. Traditional algorithms require attaching $O(n^2)$ or $O(n)$ causal histories to every application message, severely degrading network throughput \cite{raynal1991causal, cykas, KS}. While recent advancements have reduced this ov erhead, achieving constant $O(1)$ message size in a purely asynchronous, Byzantine-tolerant model has remained elusive.

The evolution of causal delivery has yielded distinct solutions tailored for broadcast \cite{birman1987reliable,Newtop,Baldoni,Anastasi}, multicast \cite{Rodrigues,KS}, and unicast \cite{raynal1991causal,cykas,Mattern,Blessing} environments. %Early fault-tolerant group communication protocols, such as Newtop, provided causal broadcast by conservatively delaying delivery until a message was received from every process in the network \cite{Newtop}. As distributed applications evolved, algorithms were adapted to support multicast environments \cite{Schiper1989}. However, 
Broadcast solutions cannot be naively applied to point-to-point unicast networks; multicast solutions if applied suffer from inefficiencies. %Utilizing a broadcast algorithm for unicast communication forces every private point-to-point message to be globally transmitted to the entire network, generating prohibitive communication overhead. Furthermore, attempting to bridge this gap by simulating a multicast via a sequence of independent unicast messages breaks causal integrity entirely; a receiver of the first simulated unicast could generate and send a new causally-dependent message that arrives at a third process before the remainder of the original simulated sequence. Thus, a 
A dedicated unicast algorithm must manage causal dependencies without relying on global message visibility or sequence simulation.
To enforce these causal dependencies, dedicated algorithms traditionally rely on either receiver-side or sender-side buffering. Receiver-buffering algorithms append causal history metadata to each message and suffer from a %massive
$O(n^2)$ metadata overhead, where $n$ is the number of processes in the system, \cite{raynal1991causal, KS}. Sender-side buffering algorithms  \cite{Mattern, cykas} suffer from non-causal latency penalties and severe throughput bottlenecks. % where 
Here {\it throughput} is the rate of computation performed or of messages handled and which depends on the number of messages in transit allowed per sender and independence from network latency between processes.
To circumvent the heavy metadata overhead of receiver-buffering without adopting the throughput limitations of sender-buffering, 
some algorithms leverage the underlying pphysical or logical topology of the network \cite{Blessing,Rodrigues} to naturally order messages in transit. Other protocols rely on specific acknowledgment patterns and blocking communication structures to maintain consistency \cite{Baldoni}, which can inadvertently introduce hidden liveness issues and deadlocks \cite{Anastasi}. %In this paper, we explicitly avoid these 
Such topology-dependent and blocking designs %. By remaining strictly topology-agnostic, our algorithm prevents the creation of 
create centralized network bottlenecks and incur %avoids the 
complex reconfiguration overhead when nodes fail. %, making %it highly
%them unsuitable for scalable and dynamic decentralized environments.

Guaranteeing causal delivery in an asynchronous system subject to Byzantine failures remains a challenge. In a Byzantine fault model, compromised processes can act arbitrarily, dropping messages, forging causal histories, or executing resource-exhaustion attacks by withholding required control messages. Recent theoretical results have proven that achieving strict causal ordering (strong safety) alongside guaranteed message delivery (liveness) without using cryptography is mathematically impossible in an  asynchronous Byzantine system  \cite{DBLP:journals/tpds/MisraK24,DBLP:conf/icdcn/MisraK23,DBLP:conf/netys/MisraK22}. Consequently, distributed systems must trade strong safety for liveness, or rely on synchrony assumptions such as bounded latency.
The existing deterministic Byzantine-tolerant causal ordering algorithm for unicasts in asynchronous systems uses cryptography and suffers from high message space overheads, degrading network throughput \cite{minicast}. While recent advancements for crash-tolerant systems have low message overhead \cite{almeida,cykas}, achieving constant $O(1)$ message size in a purely asynchronous, Byzantine-tolerant model has remained elusive.
\begin{comment}
Furthermore, existing causal ordering algorithms typically suffer from high message space overheads. Traditional algorithms require attaching $O(n^2)$ or $O(n)$ causal histories to every application message, degrading network throughput \cite{raynal1991causal, cykas, KS}. While recent advancements have reduced this overhead \cite{almeida}, achieving constant $O(1)$ message size in a purely asynchronous, Byzantine-tolerant model has remained elusive.
\end{comment}

\subsection*{Contributions}
In this paper, we present a novel $O(1)$-message overhead Byzantine-tolerant causal ordering algorithm. Our approach builds upon the space-optimal Sender Permission to Send (SPS) invariant of Almeida  \cite{almeida}, but fundamentally redesigns the memory management architecture into an \textit{Isolated-Buffer Optimistic Model}. To circumvent Byzantine state-pinning and head-of-line blocking attacks, we replace unified causal buffers with isolated, per-peer queues equipped with event-driven \textit{Cascading Space Evictions}. By managing these isolated memory bounds completely locally, our protocol eliminates the dependency on heavy, multi-round cryptographic consensus primitives or global broadcast abstractions. Instead, the entire architecture operates over a single point-to-point interconnect utilizing efficient unicast primitives for all data and signaling plane operations, including application payloads ($\text{MSG}$), acknowledgments ($\text{ACK}$), and authorization clearances ($\text{PERMIT}$). 

We formally prove that our protocol guarantees deterministic system-wide liveness for honest-to-honest communications in a purely asynchronous environment while satisfying the \textit{Congestion-Relaxed Causal Delivery} abstraction. Under this model, causal safety is strictly preserved for all honest-to-honest communications unless extreme network latency or Byzantine withholding attacks exceed quantifiable local buffer capacities, forcing optimistic queue evictions. Ultimately, this structural paradigm successfully balances resilient Byzantine fault tolerance with optimal network message sizes, minimal computational costs, and strictly bounded $O(n^2)$ local space complexity.

The core complexity metrics of our process-isolated algorithm are characterized by the following properties:

\begin{itemize}
    \item \textbf{Unified Point-to-Point Unicast:} In contrast to existing asynchronous Byzantine-tolerant causal frameworks that are mathematically restricted to broadcast primitives \cite{DBLP:journals/tcs/AuvolatFRT21, Cachin2001} or protocols that must enforce strict synchronous network latency bounds ($\delta$) and timeout-driven designs to guarantee progress \cite{DBLP:journals/tpds/MisraK24}, our architecture operates entirely over a single, uniform asynchronous point-to-point network. Crucially, we eliminate the need for separate or specialized messaging infrastructures by routing both application data payloads ($\text{MSG}$) and control signaling primitives ($\text{ACK}$, $\text{PERMIT}$) as point-to-point unicast transmissions over this exact same underlying network layer. This unified peer-to-peer approach bypasses synchronous latency assumptions and global consensus channels entirely, ensuring low latency, high throughput, and robust structural process isolation while keeping both data and control communication fundamentally private to the sender and receiver.
    \item \textbf{Network Message Space Overhead ($O(1)$):} The protocol eliminates the traditional requirement of piggybacking linear $O(n)$ or quadratic $O(n^2)$ history metadata onto network transmissions. By restricting the Dependency Vectors ($\mathbf{DV}$) and outbound history snapshots ($\mathbf{LMS}$) entirely to internal, process-local state workspaces, application packets carry only a fixed set of scalar components, preserving a strict $O(1)$ message size overhead with respect to the network size $n$.
    \item \textbf{Local Space Complexity ($O(n^2)$):} To protect against cross-channel Byzantine denial-of-service and buffer-exhaustion attacks, local data management is partitioned into process-segregated unacknowledged buffers ($\mathbf{U}_i[1..n]$) and missing permit buffers ($\mathbf{Q}_i[1..n]$), bounded by capacity thresholds $M_{\max}$ and $Q_{\max}$, respectively. Because each cached message entry encapsulates its own $n$-sized dependency array ($m.\mathbf{DV}$) and outbound tracking snapshot ($m.\mathbf{LMS}$), the total local space complexity scales as $O(n^2)$. This quadratic memory commitment represents the mathematically necessary cost to achieve complete fault isolation and $O(1)$ network message sizes.
    \item \textbf{Computational Time Complexity (($O(n^2)$ Worst-Case Pass / Amortized $O(n)$):} Basic enqueue, dequeue, and lookup operations within the process-isolated queues execute in constant $O(1)$ time relative to network scale. Causal evaluations and permit-clearance verification routines (\textit{sendMessages} and \textit{sendPermits}) perform localized vector comparisons of size $n$, running in worst-case $O(n^2)$ time over the complete array of channels. However, because these passes are event-driven, bounded by capacity limits, and execute only upon the generation or ingestion of state-transition signals, the computational overhead is highly optimized and amortized $O(n)$ across high-throughput application pipelines.
\end{itemize}

\noindent{{\bf Outline:}} 
Section~\ref{sec:background} gives the system model and background. Section~\ref{sec:prevwork} reviews related work. Section~\ref{sec:algo} gives our algorithm. Section~\ref{sec:proof} gives the correctness proof. Section~\ref{sec:complexity} gives the complexity analysis. %of our algorithm.
Section~\ref{sec:conc} gives a discussion and the conclusions.

\section{System Model and Background}
\label{sec:background}

\subsubsection*{System Model}
We assume an asynchronous distributed system comprising of $n$ processes. The system is  subject to Byzantine failures. Any $t$ number of processes, where $t < n$, could behave arbitrarily, by dropping messages, forging messages, or colluding with other Byzantine processes.
The system has reliable FIFO communication channels between each pair of processes for application messages (e.g. TCP). %It also has an Asynchronous Byzantine Atomic Broadcast (ABAB) channel for control messages. Because deterministic atomic broadcast is impossible in a purely asynchronous system under failures, we assume the ABAB primitive is implemented using a randomized Validated Asynchronous Byzantine Agreement (VABA) protocol \cite{abraham}. Specifically, we assume an asymptotically optimal VABA protocol, such as the protocol by Abraham et al. \cite{abraham}, which guarantees global total order with an expected $O(1)$ running time and $O(n^2)$ expected word communication. 

\subsubsection*{Causal Ordering and Safety Definitions}
Causal ordering \cite{birman1987reliable} is predicated on the \textit{happens-before} relation ($\rightarrow$) on the set of events formulated by Lamport \cite{LLclock}. 
%If the send event of message $m_1$ happens before the send event of $m_2$, causal ordering dictates that $m_1$ must be delivered before $m_2$ at any common destination. This property is known as \textit{strong safety}. 
Under this framework, the causal history of any given message $m$ in a set of application data transmissions $R$ is defined by its \textit{causal past}, denoted as $CP(m) = \{m' \in R \mid m' \rightarrow m\}$, which encompasses all messages that causally precede it and where $m' \rightarrow m$ is defined by $send(m') \rightarrow send(m)$ \cite{DBLP:journals/tpds/MisraK24}. 
A causal delivery protocol satisfies \textit{Strong Safety} if, for any two messages where $m' \in CP(m)$ destined for a common correct process, the receiver never delivers $m$ before $m'$. 

\subsubsection*{The Impossibility of Cryptography-Free Deterministic Byzantine Causal Ordering}
In a purely asynchronous system, there is no fixed upper bound on message latency or processor speed. In such an environment, the Fischer-Lynch-Paterson (FLP) impossibility result states that deterministic consensus is unsolvable under failures \cite{DBLP:journals/jacm/FischerLP85}. Building upon this, Misra and Kshemkalyani proved that cryptography-free deterministic Byzantine-tolerant causal ordering ensuring both strong safety and liveness is impossible for unicasts, multicasts, and broadcasts \cite{DBLP:journals/tpds/MisraK24,DBLP:conf/icdcn/MisraK23}. Thus, any such asynchronous Byzantine-tolerant causal algorithm must either accept localized strong safety violations or risk permanently halting honest processes (sacrificing liveness).

\subsubsection*{Adapting Causal Ordering to Byzantine Systems}
In a Byzantine environment, a Byzantine process can manipulate its local execution history. It may receive $m_1$, generate $m_2$ based on that information, but send $m_2$ while denying it ever received $m_1$. Internal process state is opaque to the network, and this specific omission cannot be prevented (using a deterministic solution) \cite{DBLP:journals/tpds/MisraK24,DBLP:conf/icdcn/MisraK23,DBLP:conf/netys/MisraK22}. 
%Consequently, literature defines a secondary property known as \textit{weak safety}, which requires that $m_1$ is delivered before $m_2$ only if the causal path from $m_1$ to $m_2$ passes exclusively through correct processes \cite{DBLP:journals/tpds/MisraK24}.
To establish an ordering boundary that remains mathematically achievable despite malicious historical omissions, the system restricts focus to the set $S \subseteq R$ of messages sent and delivered exclusively by correct processes, defining the \textit{Byzantine happens-before} relation ($\stackrel{B}{\rightarrow}$). 
\begin{definition}
The Byzantine happens before relation $\stackrel{B}{\rightarrow}$ on
messages in $S$ consists of the following rules \cite{DBLP:journals/tpds/MisraK24}:
\begin{enumerate}
\item If $p_i$ is a correct process and $p_i$ sent or delivered message $m$ to/from another correct process before sending message $m'$ to a correct process, then $m \stackrel{B}{\rightarrow} m'$.
\item If $m \stackrel{B}{\rightarrow} m'$ and $m' \stackrel{B}{\rightarrow} m''$, then $m \stackrel{B}{\rightarrow} m''$.
\end{enumerate}
\end{definition}
This relation filters out causal pathways passing through faulty nodes, constructing the \textit{Byzantine causal past}, denoted as $BCP(m) = \{m' \in S \mid m' \stackrel{B}{\rightarrow} m\}$\cite{DBLP:journals/tpds/MisraK24}. Consequently, 
a causal delivery protocol satisfies \textit{Weak Safety} if, for any two messages where $m' \in BCP(m)$ destined for a common correct process, the receiver never delivers $m$ before $m'$ 
%a protocol satisfies \textit{Weak Safety} if it guarantees that no correct process delivers a message $m$ before any predecessor $m' \in BCP(m)$ 
\cite{DBLP:journals/tpds/MisraK24}.

\begin{comment} 
To defend against active tampering rather than just omission, algorithms must also enforce \textit{cryptographic integrity}. In this architecture, cryptographic integrity guarantees that neither the application payload nor the attached causal metadata (specifically, the $O(1)$ hash tip and ledger height) can be forged, altered, or illegitimately replayed by a Byzantine process. It ensures that if a correct process verifies and accepts a message, the message provably originated from the claimed sender, and its attached causal proof represents a mathematically immutable snapshot of the sender's state at the exact moment of transmission.
\end{comment}

To enforce these causal restrictions without piggybacking variable-sized metadata arrays or linear dependency histories onto transmitted payloads, our protocol incorporates the \textit{Sender Permission to Send (SPS)} invariant introduced by Almeida \cite{almeida}. Rather than tracking an explicit causal past on messages, Almeida defines the SPS invariant temporally: a process $i$ is restricted from network-sending an application message $m$ to a process $j$ when there are still undelivered messages that were sent to any process other than $i$ by any peer process $k \neq j$ that previously transmitted messages to $i$ \cite{almeida}. This strategy ensures that for any message delivered by $i$ prior to creating $m$, the sender of that message has already received acknowledgments for its own historical dependencies \cite{almeida}. Assuming an underlying communication fabric that guarantees point-to-point FIFO delivery, the universal enforcement of this invariant across all correct nodes is mathematically sufficient to guarantee system-wide causal delivery ($\text{SPS} + \text{FIFO} \implies \text{Causal}$). Our protocol adapts this structural invariant to a Byzantine setting, leveraging process-isolated outbound queues ($\mathbf{U}_i$) and missing permit buffers ($\mathbf{Q}_i$) to clear authorization bounds while ensuring fault containment and network liveness.

Crucially, the algorithm separates the local scheduling of a message from its transmission into the network by distinguishing between an application-initiated \textit{causal-send} event and an algorithm-controlled \textit{network-send} event. When an application executes \textit{causalSend}, the routine immediately snapshots local tracking states and returns control to the execution thread to maximize concurrency, but intercepts the message record and buffers it locally. The physical network-send over the point-to-point network is then delayed until the underlying signaling plane verifies that the release of $m$ is safe to execute without violating causal dependencies.

\subsubsection*{Liveness}
Alongside the safety properties, causal ordering algorithms must also ensure the property of \textit{liveness}. In the presence of Byzantine faults, a weakening of standard liveness is required: liveness guarantees that every message sent by a correct process to another correct process will eventually be delivered.

\subsubsection*{The Congestion-Relaxed Causal Delivery Abstraction}
Because strictly enforcing causal safety is impossible in asynchronous Byzantine networks without risking deadlock, we define a relaxed safety abstraction called \textit{Congestion-Relaxed Causal Delivery}. %Let $K = M_{\max} + Q_{\max}$ represent the capacity bound of local process queues. Under this abstraction:
Let $M_{\max}, Q_{\max}$ represent the capacity bound of local $\mathbf{U}$, $\mathbf{Q}$, respectively, process queues. Further, $m.per$ is a flag that indicates whether a PERMIT is required for message $m$, and $m.mid$ is $m$'s sender-assigned sequence number. Under this abstraction:

\begin{definition}[Congestion-Relaxed Causal Delivery (CRCD)]
\label{def:crcd}
A system satisfies Congestion-Relaxed Causal Delivery if:
\begin{enumerate}
    \item \textbf{FIFO Order}: Messages sent directly from process $i$ to process $j$ are delivered in the order sent.
    \item \textbf{Bounded Causal Integrity}: 
    For any two messages $m_{i_0}, m_{i_x}$ sent by correct processes $i_1$ and $i_x$, respectively, to a common correct process $i_0 = i_{x+1}$, where $m_{i_0} \stackrel{B}{\rightarrow} m_{i_x}$ along causal path $m_{i_1}, m_{i_2},\ldots m_{i_x}$ ($x \ge 2$) sent by correct processes $i_1, i_2,\- \ldots i_x$, respectively, 
    if:
    \begin{enumerate}
        \item process $i_1$, the sender of $m_{i_0}$ to $i_0$, has {\em causalSend} less than $M_{\max}$ messages to $i_0$ since enqueuing $m_{i_0}$ in $\mathbf{U}_{i_1}[i_0]$ until $m_{i_0}$ gets dequeued, and %since inserting $m_{i_0}$ in $U_{i_1}[i_0]$ until dequeuing it from $Q_{i_1}[i_0]$, and
        \item $\neg m_{i_1}.per$ or process $i_1$, the sender of $m_{i_1}$ to $i_2$, has {\em causalSend} less than $Q_{\max}$ messages to $i_2$ since enqueuing $m_{i_1}$ in $\mathbf{Q}_{i_1}[i_2]$ until $m_{i_1}$ gets dequeued, and
        \item $\neg m_{i_1}.per$ or process $i_2$, the sender of $m_{i_2}$ to $i_3$, has received $PERMIT(m_{i_1})$ from $i_1$ before $i_2$ network-sends $m_{i_2}$
    \end{enumerate}
    then $i_0 (= i_{x+1})$ never delivers $m_{i_x}$ before $m_{i_0}$.
\end{enumerate}
\end{definition}

This abstraction formally characterizes the observable safety guarantee: weak causal safety is violated only when a process experiences extreme network congestion or adversarial stalling that forces local queue eviction beyond the bounds stated in the characterization. If the  bounds are never exceeded, the abstraction reduces to standard Weak Safety.

Definition~\ref{def:crcd} is for causal ordering between a specific $m_{i_0}$ and $m_{i_x}$. Observe that this definition is equivalent to the following generalization.
\begin{definition}[Generalized Congestion-Relaxed Causal Delivery]
\label{def:crcd-gen}
A system satisfies Congestion-Relaxed Causal Delivery if:
\begin{enumerate}
    \item \textbf{FIFO Order}: Messages sent directly from process $i$ to process $j$ are delivered in the order sent.
    \item \textbf{Bounded Causal Integrity}: 
    For all correct processes $dst$,
    %For any two messages $m_{i_0}, m_{i_x}$ sent by correct processes $i_1$ and $i_x$, respectively, to a common correct process $i_0 = i_{x+1}$, where $m_{i_0} \stackrel{B}{\rightarrow} m_{i_x}$ along causal path $m_{i_1}, m_{i_2},\ldots m_{i_x}$ ($x \ge 2$) sent by correct processes $i_1, i_2,\- \ldots i_x$, respectively, 
    if at each correct process $i$:
    \begin{enumerate}
        \item there is no eviction from $\mathbf{U}_i[dst]$, and
        %\item process $i_1$, the sender of $m_{i_0}$ to $i_0$, has {\em causalSend} less than $M_{\max}$ messages to $i_0$ since enqueuing $m_{i_0}$ in $\mathbf{U}_{i_1}[i_0]$ until $m_{i_0}$ gets dequeued, and %since inserting $m_{i_0}$ in $U_{i_1}[i_0]$ until dequeuing it from $Q_{i_1}[i_0]$, and
        \item there is no eviction from $\mathbf{Q}_i[dst]$, and
        %\item $\neg m_{i_1}.per$ or process $i_1$, the sender of $m_{i_1}$ to $i_2$, has {\em causalSend} less than $Q_{\max}$ messages to $i_2$ since enqueuing $m_{i_1}$ in $\mathbf{Q}_{i_1}[i_2]$ until $m_{i_1}$ gets dequeued, and
        \item before network-sending $m$ to $dst$, an explicit $PERMIT$ is received from  correct $j$ for $m'$ (sent by $j$), where $m'.mid = m.\mathbf{DV}[j]$ and $m'.per=1$, 
        %\item $\neg m_{i_1}.per$ or process $i_2$, the sender of $m_{i_2}$ to $i_3$, has received $PERMIT(m_{i_1})$ from $i_1$ before $i_2$ network-sends $m_{i_2}$
    \end{enumerate}
    then weak causal safety is satisfied.
    %then $i_0 (= i_{x+1})$ never delivers $m_{i_x}$ before $m_{i_0}$.
\end{enumerate}
\end{definition}
This definition stipulates that CRCD is satisfied if (at correct processes) there is no eviction from the queues for correct processes and all due permissions (PERMITs) from correct processes arrive before network-sending messages to correct processes.

%\subsubsection*{Asynchronous Byzantine Atomic Broadcast (ABAB)}
%To maintain a globally consistent view of causal dependencies without attaching $O(n)$ vectors to every message, our algorithm relies on an Asynchronous Byzantine Atomic Broadcast (ABAB) channel for control messages. While deterministic atomic broadcast is impossible under FLP, modern ABAB protocols bypass this limitation using randomized or threshold-cryptographic consensus mechanisms. These protocols guarantee total order, ensuring that all correct processes eventually deliver the exact same sequence of control events, which allows our algorithm to compute a deterministic, constant-sized running Hash Ledger.

\section{Previous Work}
\label{sec:prevwork}
\subsubsection*{Crash-Tolerant Causal Delivery}
The Raynal-Schiper-Toueg (RST) algorithm was among the first to solve causal ordering but required a space and message overhead of $O(n^2)$, as every message carried the sender's local matrix of known message states \cite{raynal1991causal}. Subsequent variations, such as the KS algorithm \cite{KS}, improved practical performance but maintained the $O(n^2)$ worst-case bound. More recently, the Cykas algorithm \cite{cykas} adopted a sender-buffering approach that reduced the message metadata overhead to $O(1)$, though it suffers from throughput scalability and liveness limitations. 

Almeida subsequently introduced a hybrid-buffering approach utilizing the Sender Permission to Send (SPS) invariant, which successfully provides liveness and throughput scalability and achieves space and computational optimality \cite{almeida}. Under the SPS invariant, causality is enforced at the sender side rather than the receiver; a process is mathematically restricted from transmitting a new message over the network until it receives confirmation (a \textit{PERMIT}) that all causally preceding messages have been successfully delivered to their respective destinations. Almeida's algorithm implements this by combining sender-buffering to delay network transmissions until the SPS conditions are met, and receiver-buffering to enforce strict FIFO ordering and allow message pipelining. 

To achieve an amortized $O(1)$ computational complexity, the algorithm completely avoids scanning message histories. Instead, it utilizes specialized data structures—namely sliding arrays and indexable sliding maps—to perform constant-time tracking of pending \textit{PERMIT} and missing acknowledgments (\textit{ACK}s). These \textit{ACK}s serve as the critical feedback loop for the SPS invariant: when a receiver successfully delivers a message, it replies with an \textit{ACK}. This allows the original sender to clear the message from its unacknowledged sliding array and subsequently issue the required \textit{PERMIT}s to unblock causally dependent messages across the network. 

However, Almeida's unified buffering structures and implicit trust in the faithful generation of \textit{ACK}s and \textit{PERMIT}s are designed for crash-fault models, making them highly vulnerable to resource-exhaustion and state-pinning attacks in a Byzantine environment.

\begin{table*}[t]
\begin{center}
{\small
\begin{minipage}{\textwidth}
\begin{center}
\begin{tabular}{|l|l|l|l|l|l|l|l|l|}
%\toprule %
\hline
 & Cachin et  & Auvolat  & Sender-In-  & Channel & Minicast \cite{minicast} & Patel & Patel & {\bf Proposed}  \\
 & al. \cite{Cachin2001} & et al. \cite{DBLP:journals/tcs/AuvolatFRT21} & hibition \cite{DBLP:journals/tpds/MisraK24} & Sync \cite{DBLP:journals/tpds/MisraK24} &  &  et al. \cite{PKbcrb} & et al. \cite{PKbcrb2} & \\
%\midrule %
\hline \hline

Comm. mode & Broadcasts & Broadcasts & Unicasts & Unicasts & Unicasts & Broadcasts & Broadcasts & Unicasts \\
\hline

Asynchronous & \cmark & \cmark & \xmark & \xmark & \cmark & \cmark & \cmark & \cmark   \\
\hline

Deterministic & \xmark & \cmark & \cmark & \cmark & \cmark & \cmark & \cmark & \cmark    \\
\hline

Throughput- & \xmark & \cmark & \xmark & \cmark & \xmark & \cmark & \cmark & \cmark \\
scalable & & & & & & & & \\
\hline

Crypto-free & \xmark & \cmark & \cmark & \cmark & \xmark & \xmark & \xmark & \cmark \\
\hline

Liveness & \cmark & \cmark & \cmark & \cmark & \cmark & \cmark & \cmark & \cmark \\
\hline

%Safety & \cmark & \xmark & \xmark & \xmark & \xmark & \xmark \\
Safety & $Pr(\text{violation})$ & Weak safety & Weak safety & Weak safety & Weak safety & $Pr(\text{violation})$ & Strong & C.-R. Causal \\
guarantee & $ = \epsilon > 0$ &  & & & & $ = O(n^{-1} \ln^2 n)$ & safety & Delivery \\
\hline

Msg. space & $O(n)$ & $O(n)$ & $O(1)$ & $O(1)$ & $O(n)$ & $O(1)$ & $O(1)$ & $O(1)$ \\
overhead & & & & & & & & \\
\hline

Number of msgs & $O(n^2)$ of $O(n +$-  & $O(n^2)$ of & 1 of $O(1)$ $+$ & $2n$ of $O(1)$ $+$ & $O(n^2)$ of $O(1)$ $+$ & $O(n^2)$ of & $O(n^2)$ of & 2 of $O(1)$ $+$\\
 & $\text{encrypted} \; |m|)$ & $O(n+|m|)$ & 1 of $O(|m|)$ & 1 of $O(|m|)$ & 1 of $O(n+|m|)$ & $O(|m|)$ & $O(|m|)$ & 1 of $O(|m|)$ \\
\hline

Communication  & $O(n^2(n+|m|))$ & $O(n^2\cdot$ & $O(|m|)$ & $O(n + |m|)$ & $O(n^2 + |m|)$ & $O(n^2|m|)$ & $O(n^2|m|)$ & $O(|m|)$ \\
word complexity & & $(n+|m|))$ & & & & & & \\
\hline

Local space & unbounded & unbounded & $O(1)$ & $O(n)$ & unbounded & unbounded & unbounded & $O(n^2)$ \\
% & at all processes & at all processes & & at all processes & at all processes & \\
\hline

Time & $O(n)$ & $O(n)$ & $O(1)$ & $O(n)$ & $O(n)$ & $O(n)$ & $O(n)$ & $O(n)$ \\
complexity & at all procs & at all procs & & at all procs & at all procs & & & amortized \\
\hline

%\bottomrule 
\end{tabular}
\end{center}
\end{minipage}
}
\end{center}
%\end{minipage}
%\medskip
\caption{Comparison of Byzantine-tolerant causal ordering algorithms for unicast. Message is denoted $m$. 
}
\label{table:comparison}
\end{table*}

\subsubsection*{Byzantine-Tolerant Causal Ordering}
%\subsubsection*{Validated Asynchronous Byzantine Agreement}
%Because our control plane relies on ABAB to sequence control messages, the efficiency of our algorithm is bound by the efficiency of the underlying broadcast primitive. 
Cachin et al. initially formalized Validated Asynchronous Byzantine Agreement (VABA), achieving optimal resilience and $O(1)$ time, but requiring $O(n^3)$ expected word communication \cite{VABACachin} and using randomization. 
In a separate work focusing on secure service replication, Cachin et al. proposed a Secure Causal Atomic Broadcast protocol built directly atop this atomic broadcast primitive \cite{Cachin2001}. To guarantee strict ``secure causality'', preventing an adversary from front-running or manipulating causal dependencies based on message contents, their protocol employs threshold encryption. A sender encrypts the full application payload, and the ciphertext is routed through the atomic broadcast network. Only after the encrypted message is globally sequenced do the participating nodes exchange decryption shares to reveal the plaintext payload. This robustly prevents information laundering and input-causality attacks but tightly couples the data and control planes. Consequently, the heavy expected $O(n^3)$ communication complexity of the consensus layer must process the full application payload size $O(|m|)$, while the network simultaneously absorbs the overhead of distributing threshold decryption shares. As the underlying atomic broadcast uses randomized consensus, there is $\epsilon > 0$ probability of causal ordering violation.

Two protocols for crash failures and a third for Byzantine failure of clients using cryptography were proposed for secure causal atomic broadcast \cite{DBLP:conf/dsn/DuanRZ17}; the third made assumptions on latency of messages, and hence works only in a synchronous system.

The Byzantine causal broadcast algorithm of Auvolat et al. \cite{DBLP:journals/tcs/AuvolatFRT21} uses the BRB primitive \cite{DBLP:journals/iandc/Bracha87} ($O(n^2)$ messages) over the message payloads having $O(n)$ message space overhead, and cannot guarantee strong safety. It also uses unbounded memory. % despite using $O(n)$ message space overhead. 

Due to the impossibility of solving Byzantine causal ordering deterministically without using cryptography in asynchronous systems \cite{DBLP:journals/tpds/MisraK24}, recent practical solutions have relied on strengthening the system model with synchronous assumptions (bounded latency). %Misra and Kshemkalyani proposed 
The \textit{Sender-Inhibition} and \textit{Channel Sync} algorithms, which guarantee weak safety and liveness for unicasts by utilizing timeouts, use $O(1)$-sized $O(1)$ and $O(n)$ control messages, respectively \cite{DBLP:journals/tpds/MisraK24}. 

For an asynchronous system, Rumreich and Sivilotti introduced the \textit{Minicast Algorithm}, which reduces the worst-case message size complexity to $O(n)$ by broadcasting small control messages alongside application unicasts \cite{minicast}. By using a Byzantine Reliable Broadcast (BRB) primitive \cite{DBLP:journals/iandc/Bracha87} with FIFO ordering ($O(n^2)$ messages), their solution achieves Byzantine tolerance but inherits an $O(n)$ message space overhead. Moreover, it cannot guarantee strong safety. Crucially, however, their design assumes unbounded local space. As the algorithm lacks strict capacity limits for its receiving queues, a single Byzantine process can execute a memory exhaustion attack by flooding an honest process with invalid application messages, causing its local memory to grow unboundedly.

Patel et al. \cite{PKbcrb} proposed four $\mathcal{O}(1)$ metadata overhead algorithm variants for Byzantine Causal Reliable Broadcast, all having $\mathcal{O}(1)$ message metadata. By modeling link propagation delays as independent exponential variables, they derived bounds on the probability of weak safety violations and of strong safety violations. Patel et al. \cite{PKbcrb2} also  proposed a $\mathcal{O}(1)$ metadata overhead algorithm for Byzantine Causal Reliable Broadcast guaranteeing strong safety.

\begin{comment}
    
In contrast, our Isolated-Buffer Optimistic architecture routes application payloads directly over lightweight, point-to-point FIFO channels and intentionally accept mathematically unavoidable information laundering attacks at the local application layer. This architectural tradeoff keeps the actual data payloads completely off the expensive consensus channel. 
\end{comment}
%In contrast, our Isolated-Buffer Optimistic architecture explicitly decouples these planes. We route application payloads directly over lightweight, point-to-point FIFO channels and intentionally accept mathematically unavoidable information laundering attacks at the local application layer. This architectural tradeoff %allows us to keep keeps the actual data payloads completely off the expensive consensus channel. 
%
%Abraham et al. subsequently resolved the VABA consensus bottleneck, providing a protocol with asymptotically optimal communication—specifically, an expected $O(n^2)$ messages of $O(1)$ size each and $O(1)$ expected running time \cite{abraham}. By leveraging this optimal VABA implementation strictly for our constant-sized control metadata (\textit{PERMIT}s and \textit{ACK}s), our ABAB control plane ensures that the global causal ledger can be securely maintained without the application payloads inflating the consensus layer's network complexity.
%
A comparison of Byzantine-tolerant causal delivery algorithms for unicasts is given in Table~\ref{table:comparison}. Note, the broadcast algorithms, if used for unicast, require an additional level of encryption with the recipient's public key to maintain privacy.

\section{Proposed Algorithm}
\label{sec:algo}
\begin{comment}
\subsection{System Model and Assumptions}
We assume an asynchronous distributed system comprising of $n$ processes. The system has reliable FIFO communication channel for application messages (e.g. TCP). It also has an Asynchronous Byzantine Atomic Broadcast (ABAB) channel for control messages. Because deterministic atomic broadcast is impossible in a purely asynchronous system under failures, we assume the ABAB primitive is implemented using a randomized Validated Asynchronous Byzantine Agreement (VABA) protocol \cite{abraham}. Specifically, we assume an asymptotically optimal VABA primitive, such as the protocol by Abraham et al., which guarantees global total order with an expected $O(1)$ running time and $O(n^2)$ expected word communication. 

The system is also subject to Byzantine failures. Any $t$ number of processes could behave arbitrarily, by dropping messages, forging messages, or colluding with other Byzantine processes. The Byzantine processes are limited to $t < n/3$. 
\end{comment}
As our architectural tradeoff chooses not to use cryptography and randomization in the data plane to keep encrypted data payloads off the expensive consensus channel (consensus was run on encrypted data in \cite{Cachin2001}), and uses a lightweight deterministic solution (that requires only 
the $O(1)$-sized message space overhead and $O(1)$ control messages at the cost of $O(n^2)$ local space), the result of \cite{DBLP:journals/tpds/MisraK24} is applicable.
As it is not possible to guarantee strong safety and liveness deterministically and without using cryptography in a Byzantine system \cite{DBLP:journals/tpds/MisraK24}, our algorithm prioritizes liveness. It allows some occasional violations of safety to ensure liveness and simultaneous O(1) message overhead and $O(n^2)$ local space overhead. In scenarios of extreme latency or malicious resource-exhaustion attacks, the system gracefully degrades to allowing occasional violations of safety to prevent causal deadlocks; however {\em Congestion-Relaxed Causal Delivery} is always guaranteed. 

\subsection{Core Logic}
We build on the work done by Almeida and we introduce Byzantine tolerance to his Sender Permission to Send invariant \cite{almeida}. %Almeida's classic protocol achieves space and computational optimality, topology agnosticism, and scalable throughput by leveraging a hybrid architecture of sender and receiver side buffering. He introduced the Sender Permission to Send (SPS) invariant and utilized it in his causal delivery algorithm. 
%
%Our approach enhances the network model to tolerate up to 
However, deploying such a model within a purely asynchronous Byzantine environment introduces critical structural vulnerabilities that must be systematically mitigated to preserve both liveness and causal integrity:

\begin{itemize}
    \item \textbf{Head-of-Line Blocking and State Pinning:} Byzantine processes can maliciously withhold required control messages (such as \textit{PERMIT}s or \textit{ACK}s) to permanently halt an honest process's outgoing communication pipeline.
    \item \textbf{Active Causal Forgery:} Byzantine processes may falsify their causal history to release messages that illegally bypass the Sender Permission to Send (SPS) invariant rules.
    \item \textbf{Memory Exhaustion:} Byzantine processes can flood the network with a stream of invalid or artificially future-dated application messages to consume unbounded local memory at honest processes.
    \item \textbf{Message Replay Attacks:} Byzantine processes may capture and repeatedly replay historically valid messages to trigger duplicate deliveries or corrupt the local causal state.
\end{itemize}

%While the previous asynchronous Byzantine-tolerant approach --- the Minicast algorithm, attempted to mitigate such Byzantine behaviors, it still incurs an $O(n)$ message size complexity by appending vectors of delivered control messages to every application payload. 
To maintain a constant $O(1)$ message size over the network, we utilize isolated Buffers ($\mathbf{U_i}$) and Dependency Vectors ($\mathbf{DV}$). Although having those local data structures results in an $O(n^2)$ local space complexity, it ensures that the messages sent over the network have a constant size. Additionally, to prevent the expensive cryptographic consensus or global broadcasts from degrading application throughput, all communication relies on a point-to-point unicast for all messages. This includes all application messages, as well as the control messages \textit{ACK} and \textit{PERMIT}. 

\subsection{Important Data Structures}
To guarantee a $O(1)$ network message size overhead and an amortized $O(n)$ computational time complexity under a Byzantine-tolerant paradigm, the proposed framework transitions away from the unified data structures utilized in Almeida's protocol \cite{almeida}. Instead, we introduce a process-isolated queue architecture designed to structurally confine the impact of adversarial behaviors, paired with a localized tracking array paradigm. The structural definitions are detailed in Algorithm~\ref{alg:byzantine_sps_fifo_gc_types}.
Their description and usage is given in Appendix~\ref{ap:ds}.

\subsection{Byzantine Tolerant Algorithm}
% Our proposed algorithm builds upon the Sender Permission to Send (SPS) invariant, but it has completely redesigned memory management into an Isolated-Buffer with optimistic evacuation. It is an event-driven algorithm that reacts to incoming point-to-point application messages and globally ordered control messages. The objective of this algorithm is to introduce Byzantine tolerance by utilizing local Dependency Vectors (DV) to verify the SPS invariant while maintaining a constant O(1) message overhead. By utilizing $O(n^2)$ local space for isolated queues $(U_i)$, the algorithm guarantees liveness in an asynchronous system by allowing occasional violations of causal safety through the eviction of stale messages.

Our proposed algorithm (Algorithm~\ref{alg:byzantine_sps_fifo_gc}) builds upon the Sender Permission to Send (SPS) invariant, introducing an event-driven framework that operates entirely over a point-to-point network utilizing unicast primitives. The architecture relies on process-isolated buffers with automated, optimistic cascading network evictions to enforce Byzantine fault tolerance while maintaining a strict $O(1)$ message-size overhead. By allocating $O(n^2)$ local memory for isolated queues ($\mathbf{U}_i, \mathbf{Q}_i$), the protocol guarantees system liveness in an asynchronous network environment, strategically permitting localized, transient violations of causal safety when physical buffer constraints are reached. %Crucially, both application data and control signals (\textit{ACK}s and \textit{PERMIT}s) are handled uniformly as point-to-point unicast transmissions over the same underlying network layer.
The procedures are described in Appendix~\ref{ap:procs}.
% \input{Types_nocrypt}
% \input{Algo_nocrypt}

%\input{Types_temp}
%\input{Algo_temp}
%\input{Types_2Q_Revised}
%\input{Algo_2Q_Revised}
% \documentclass{article}

% % Sets standard conference margins
% \usepackage[margin=.5in]{geometry} 

% \usepackage{amsmath}
% \usepackage{amssymb}
% \usepackage[linesnumbered,ruled,vlined]{algorithm2e} 
% \usepackage{multicol} % Required for the side-by-side layout

% \begin{document}

\begin{algorithm}[th]
\small
\DontPrintSemicolon %
\SetKwComment{Comment}{// \textit{}}{}

\caption{Isolated-Buffer Optimistic Causal Delivery Types. Data structures at process $i$.}
\label{alg:byzantine_sps_fifo_gc_types}

% This shrinks the spacing between lines slightly to ensure it stays on one page
\linespread{0.95}\selectfont 

\SetInd{0.2em}{0.5em} % Tightens the indentation to save horizontal space

\textbf{types \& state} \\
 Msg: record \{ \\
    mid: $\mathbb{N}$, \Comment*{Local message sequence number}
    per: $\mathbb{B}$, \Comment*{Flag: is a permit required?}
    pl: $\mathbb{P}$, \Comment*{Application payload}
    sent: $\mathbb{B}$, \Comment*{Flag: network-transmitted?}
    DV : $\text{Array}[n] \text{ of } \mathbb{N}$, \Comment*{Last Message Recd. vector; stored in $\mathbf{U}_i$ entries per message record but not transmitted}
    LMS : $\text{Array}[n] \text{ of } \mathbb{N}$, \Comment*{Last Message Sent vector; stored in $\mathbf{U}_i$ entries per message record but not transmitted}
    
%    hash: String, \Comment*{Cryptographic hash chain tip}
%    height: $\mathbb{N}$ \Comment*{Cryptographic ledger height}
 \} \\

    $\mathbf{U}_i \gets \text{Array}[n] \text{ of } [\,]$; \Comment*{Isolated unacked messages}
    $\mathbf{Q}_i \gets \text{Array}[n] \text{ of } [\,]$; \Comment*{Isolated missing permits}
    %$\mathbf{ET}_i \gets \text{Array}[n] \text{ of } -1$; \Comment*{Eviction tracker}
    $\mathbf{LD}_i \gets \text{Array}[n] \text{ of } -1$; \Comment*{Last delivered message IDs}
    $\mathbf{LMS}_i \gets \text{Array}[n] \text{ of } -1$; \Comment*{Last message sent vector}   
    $\mathbf{LPR}_i \gets \text{Array}[n] \text{ of } -1$; \Comment*{Last Permit Recd. vector}
    $ck_i \gets 0$; \Comment*{Local clock}
    $M_{max} \gets 1024$; \Comment*{Max cap. per unacked queue, adjustable}
    $Q_{max} \gets 1024$; \Comment*{Max cap. per PERMIT queue, adjustable}
% $\mathbf{P}_i \gets \text{Array}[n] \text{ of } [\,]$; \Comment*{Isolated missing permits}
% $\mathbf{F}_i \gets \text{Array}[n] \text{ of } [\,]$; \Comment*{Cryptographic holding buffer}
% $\mathbf{LMR}_i \gets \text{Array}[n] \text{ of } 0$; \Comment*{Last Msg. Recd. vector}
% $P_{max} \gets 256$; \Comment*{Max capacity per permit queue}
% $F_{max} \gets 256$; \Comment*{Max capacity per holding queue}
% $W \gets 10000$; \Comment*{Max Window Size for GC}
% $h_i \gets 0$; \Comment*{Current Ledger Height}
% $\mathbf{H}_i \gets \text{"0"}$; \Comment*{Hash Chain tip}
% $\mathbf{L}_i \gets \text{Array}[W]$; \Comment*{Ring Buffer for Stable Hashes}

 \end{algorithm}
% \documentclass{article}

% % Sets standard conference margins
% \usepackage[margin=.5in]{geometry} 

% \usepackage{amsmath}
% \usepackage{amssymb}
% \usepackage[linesnumbered,ruled,vlined]{algorithm2e} 
% \usepackage{multicol} % Required for the side-by-side layout

% \begin{document}

\begin{algorithm*}[th]
%\small
\SetKwComment{Comment}{// \textit{}}{}
\SetKwFunction{CausalSend}{causalSend}
\SetKwFunction{BCCHBroadcast}{BCCH\_Broadcast}
\SetKwFunction{ABABBroadcast}{ABAB\_Broadcast}
\SetKwFunction{sendMessages}{SendMessages}
\SetKwFunction{sendPermits}{SendPermits}
\SetKwFunction{CheckFuture}{checkFutureBuffer}
\SetKwFunction{Receive}{receive\_i}
\SetKwFunction{BDeliv}{BCCH\_Deliver\_i}
\SetKwFunction{Deliver}{deliver}
\SetKwFunction{Periodically}{Periodically}
\SetKwFunction{ADeliv}{ABAB\_Deliver\_i}
\SetKwProg{Fn}{procedure}{:}{}
\SetKwProg{On}{on}{:}{}

\caption{Isolated-Buffer Optimistic Causal Delivery (Process $i$)}
\label{alg:byzantine_sps_fifo_gc}

% This shrinks the spacing between lines slightly to ensure it stays on one page
\linespread{0.95}\selectfont 

\begin{multicols}{2}
\SetInd{0.2em}{0.5em} % Tightens the indentation to save horizontal space

%\textbf{types \& state} \\
 % Msg: record \{rcv: $\mathbb{I}$, mid: $\mathbb{N}$, per: $\mathbb{B}$, pl: $\mathbb{P}$, sent: $\mathbb{B}$, hash: String, height: $\mathbb{N}$\}

 % $\mathbf{U}_i \gets \text{Array}[n] \text{ of } [\,]$; \Comment{Isolated unacked messages}
 % $\mathbf{P}_i \gets \text{Array}[n] \text{ of } [\,]$; \Comment{Isolated missing permits}
 % $\mathbf{F}_i \gets \text{Array}[n] \text{ of } [\,]$; \Comment{Cryptographic holding buffer}
 % $\mathbf{evicted\_mid}_i \gets \text{Array}[n] \text{ of } 0$; \Comment{Eviction tracker}
 % $ck_i \gets 0$\;
 % $M_{max} \gets 1024$; \Comment{Max capacity per unacked queue}
 % $P_{max} \gets 256$; \Comment{Max capacity per permit queue}
 % $F_{max} \gets 256$; \Comment{Max capacity per holding queue}
 % $W \gets 10000$; \Comment{Max Window Size for GC}
 % $h_i \gets 0$; \Comment{Current Ledger Height}
 % $\mathbf{H}_i \gets \text{"0"}$; \Comment{Hash Chain tip}
 % $\mathbf{L}_i \gets \text{Array}[W]$; \Comment{Ring Buffer for Stable Hashes}
 
\BlankLine
\Fn{\CausalSend{$dst, p$}}{
    \If{$\mathbf{U}_i[dst].\text{size} \ge M_{max}$}{
        % $stale\_m \gets \mathbf{U}_i[dst].\text{peek}()$\;
        
%        \If{$\neg stale\_m.\text{sent}$}{
%            \Comment{Causal Stagnation. Unpin by evicting blocking permits.}
%            \For{$j \gets 1$ \KwTo $n, j\neq i$}{
%                \If{$stale\_m.\mathbf{DV}[j] \geq \mathbf{P}_i[j][z].mid \mbox{ or either/both are } \perp$}{
%                $\textbf{P}_i[j].remove(z)$\;
%                }
%            }
        
        \Comment{Network Eviction: Free space for incoming message}
        $stale\_m \gets \mathbf{U}_i[dst].\text{remove}()$\;
        %$\mathbf{ET}_i[dst] \gets stale\_m.\text{mid}$\;
        \If{$\neg stale\_m.\text{sent}$}{
            $stale\_m.sent = \text{True}$\;
            $\text{send}_i(dst, \text{MSG}, stale\_m)$\;
            $\forall x, \mathbf{LPR}_i[x] \gets \max(\mathbf{LPR}_i[x], stale\_m.\mathbf{DV}[x])$\;
        }
        \If{%$stale\_m.\text{sent} \land$  
        $stale\_m.\text{per}$}{
            \If{$\mathbf{Q}_i[dst].size \geq Q_{max}$}{
                    $m' \gets \mathbf{Q}_i[dst].\text{remove}()$\;
                    $\text{send}_i(dst, \text{PERMIT}, m'.mid)$\;
                }
            $\mathbf{Q}_i[dst].\text{add}(stale\_m)$\;
        }
    \sendPermits{} \Comment{Unblocks pending PERMITs}
    }
    $req\_per \gets (\exists x \neq dst: \mathbf{U}_i[x].\text{size} > 0)$\;
    $m \gets \text{Msg}\{\text{mid}: ck_i, \text{per}: req\_per, \text{pl}: p, \text{sent}: \text{False}\}$\;
    $ck_i \gets ck_i + 1$\;
    $m.\mathbf{DV} \gets \mathbf{LD}_i$; $m.\mathbf{LMS} \gets \mathbf{LMS}_i$; $\mathbf{U}_i[dst].\text{add}(m)$\;
    $\mathbf{LMS}_i[dst] \gets m.mid$\;
    \sendMessages{} \Comment{Unblocks pending messages}
}

%\begin{comment}
\BlankLine
\Fn{\sendMessages{}}{
    \For{$dst \gets 1$ \KwTo $n$}{
        $k \gets 0$\;
        \While{$k < \mathbf{U}_i[dst].\text{size}$}{
            $m \gets \mathbf{U}_i[dst][k]$\;
            \If{$\neg m.\text{sent}$}{
                $blocked \gets \neg(\forall j\neq i \land j \neq dst, (m.\mathbf{DV}[j] \leq \mathbf{LPR}_i[j]))$\;
                \If{$\neg blocked$}{
                    $m.\text{sent} \gets \text{True}$\;
                    $\text{send}_i(dst, \text{MSG}, m)$\; \Comment{Exclude $m.\mathbf{DV}, m.\mathbf{LMS}$ fields from MSG}
                }
                \Else{
                    \textbf{break}\;
                }
            }
            $k \gets k + 1$\;
        }
    }
}

\BlankLine
\Fn{\sendPermits{}}{
    \For{$dst \gets 1$ \KwTo $n$}{
        \While{$\mathbf{Q}_i[dst] \neq \text{empty}$}{
            $m \gets \mathbf{Q}_i[dst].head$\;
            \If{$\forall j\neq i \land j \neq dst, (m.\mathbf{LMS}[j] < \mathbf{U}_i[j][0].mid \mbox{ or } \mathbf{U}_i[j][0] =\perp)$}{
                $\text{send}_i(dst, \text{PERMIT}, m.mid)$\;
                $\mathbf{Q}_i[dst].\text{remove}()$\;
            }
            \Else{\textbf{break}\;
            }
        }
    }
}

%\begin{comment}
\BlankLine
\Fn{\Receive{$src, \text{MSG}, m$}}{
    \If{$m.\text{mid} \le \mathbf{LD}_i[src]$}{
        \textbf{drop} $m$\;
        \Return \Comment{Drop duplicate or fake message}
    }
    \If{$\neg m.\text{per}$}{
        %\If{$\mathbf{P}_i[src].\text{size} \ge P_{max}$}{
        %    \Comment{Optimistic Eviction of oldest permit}
        %    $\mathbf{P}_i[src].\text{remove}()$\;
        %}
        $\mathbf{LPR}_i[src] \gets m.mid$\;
        %$\mathbf{P}_i[src].\text{add}(m)$\;
    }
    $\mathbf{LD}_i[src] \gets m.\text{mid}$\;
    $\text{send}_i(src, \text{ACK}, m.mid)$\;
    $\text{deliver}_i(m.\text{pl})$\;
}
%\end{comment}

% \columnbreak % --- Forces the next block into the right column ---

\BlankLine
\Fn{\Receive{$src, \text{ACK}, m\_id$}}{
    \If{$m\_id  < \mathbf{U}_i[src][0].mid$}{
        \Return \Comment{Already dequeued on ACK receipt or optimistically evicted}
    }
    \If{$m\_id > \mathbf{U}_i[src].last.mid$}{
        \Return \Comment{High numbered fake ACK}
    }
    \While{$\mathbf{U}_i[src] \neq \mbox{ empty }$}{
        \If{$\mathbf{U}_i[src][0].\text{mid} \leq m\_id$}{
            $m \gets \mathbf{U}_i[src].\text{remove}()$\;
            \If{$m.\text{sent} \land m.\text{per}$}{
                \If{$\mathbf{Q}_i[src].size \geq Q_{max}$}{
                    $m' \gets \mathbf{Q}_i[src].\text{remove}()$\;
                    $\text{send}_i(src, \text{PERMIT}, m'.mid)$\;
                }
                $\mathbf{Q}_i[src].\text{add}(m)$\;
            }
        }
        \Else{\textbf{break}\;
        }
    }
    \sendPermits{}; \Comment{Unblocks pending PERMITs}
}

\BlankLine
\Fn{\Receive{$src, \text{PERMIT}, m\_id$}}{
    \If{$\mathbf{LD}_i[src] \geq m\_id > \mathbf{LPR}_i[src]$}{
        $\mathbf{LPR}_i[src] \gets m\_id$\;
        \sendMessages{}; \Comment{Unblocks pending messages}
    }
}

\end{multicols}
\end{algorithm*}

\section{Correctness Proof}
\label{sec:proof}
%To prove the validity of our proposed algorithm, we must first establish some criteria that define correctness. 

%\subsection{Criteria}
%Traditional causal delivery algorithms strive for strong safety (if $send(m_1) \rightarrow send(m_2)$, then any process that delivers both must deliver $m_1$ before $m_2$) and liveness (if an honest process sends a message to an honest process, it will eventually be delivered).

%However, in a purely asynchronous system subject to Byzantine failures, strong safety and liveness are mathematically against each other, especially when space is bounded. A Byzantine process can exploit strong safety to permanently halt honest processes, thus violating liveness. Therefore, o
%Our algorithm does not guarantee strong safety, instead it guarantees cryptographic integrity, and liveness. The algorithm purposefully allows occasional strong safety violations, however this is done to foil any attempts by Byzantine processes to violate liveness. This is consistent with the impossibility result of \cite{DBLP:journals/tpds/MisraK24}. 

% Our algorithm does not guarantee strong safety, and instead, it guarantees liveness. The algorithm purposefully allows occasional strong safety violations. This is done to foil any attempts by Byzantine processes to violate liveness. This is consistent with the impossibility result of \cite{DBLP:journals/tpds/MisraK24}. 

Our proposed protocol does not guarantee strict, unconditioned causal safety under all execution traces; instead, it prioritizes deterministic \textit{Liveness}. As established  by the foundational impossibility result of Misra and Kshemkalyani 
\cite{DBLP:journals/tpds/MisraK24}, deterministically maintaining strong safety while guaranteeing bounded memory and unconditional liveness without using cryptography is mathematically impossible in an asynchronous system vulnerable to Byzantine faults. 
Consequently, our algorithm deliberately permits localized, transient violations of causal safety under specific adversarial conditions or extreme network asynchronous delays. This tactical trade-off is specifically engineered to neutralize state-pinning and buffer-exhaustion strategies deployed by Byzantine actors to induce system-wide deadlocks. However, Congestion-Relaxed Causal Delivery is always guaranteed.

\subsection{Liveness Guarantee}
%The primary objective of this Isolated-Buffer Optimistic algorithm is to guarantee that no Byzantine process can halt any honest process in an asynchronous system. 

\begin{theorem}[Liveness under Continuous Traffic]
\label{th:liveness}
Under the assumption that honest processes continuously generate application messages (or use a timeout-driven synchronous variant), an honest process will never be permanently blocked from sending or delivering messages, and messages sent between honest processes will eventually be delivered.
\end{theorem}
\begin{proof}
To establish this guarantee, we first trace the continuous timeline of an application message $m$ generated by a correct process $i$ and destined for another correct process $j$:

\begin{enumerate}
    \item \textbf{Buffer Admission Timeline}: When the application layer at process $i$ invokes \textit{causalSend}(j, p), the message $m$ is allocated a local sequence number $m.mid = ck_i$ and appended to the tail of the isolated application buffer $\mathbf{U}_i[j]$. If $\mathbf{U}_i[j]$ has reached its maximum capacity $M_{\max}$, the network eviction logic immediately removes the oldest entry $stale\_m$ from the head of the queue. Because $j$ is correct, any such $stale\_m$ is either immediately dispatched over the point-to-point network (if it was unsent) or, if it was already sent and requires a permit ($stale\_m.sent \wedge stale\_m.per$), it is transitioned into the missing permits queue $\mathbf{Q}_i[j]$. If $\mathbf{Q}_i[j]$ is also saturated ($\ge Q_{\max}$), a cascading eviction forces the removal of its oldest entry, immediately dispatching a unicast $\text{PERMIT}$ to $j$ to clear the channel. Thus, a correct process $i$ is never blocked from admitting a new payload into its local tracking structures.
    
    \item \textbf{Transmission Timeline}: To exit $\mathbf{U}_i[j]$ and enter the network interconnect via \textit{sendMessages}, a message $m$ must satisfy the unblocking condition: $\forall x \neq i \land x \neq j, m.\mathbf{DV}[x] \le \mathbf{LPR}_i[x]$. Because the network layer operates via reliable point-to-point channels, and because every other honest process $k$ systematically processes valid incoming payloads and returns corresponding acknowledgments, the dependency values stored within $m.\mathbf{DV}$ represent past delivered states. 
    
    If a Byzantine process $b$ maliciously suppresses control signals (such as ACKs or PERMITs) to freeze this condition, the unblocking condition may remain unsatisfied. However, if the application at process $i$ continues to generate messages to destination $j$, the queue $\mathbf{U}_i[j]$ will eventually saturate. Once its capacity reaches $M_{\max}$, the optimistic network eviction mechanism forcibly sends the head message $m$ immediately over the point-to-point network regardless of the unsatisfied blocking condition, setting $m.sent \gets \text{True}$. Similarly, if a blocked permit is held in $\mathbf{Q}_i[j]$, a cascading eviction upon queue saturation will forcibly dispatch the permit. 
    
    {\em Deadlock and Flow-Control Trade-off:} We note that this eviction-based liveness mechanism relies on the assumption that honest processes continuously generate application traffic. If the application is reactive (e.g., processes only send responses after receiving deliveries) or enforces a window-based credit flow control where the max active credits are fewer than $\min(M_{\max}, Q_{\max})$, the queues may never reach capacity to trigger evictions, resulting in a system-wide deadlock. To resolve this without traffic assumptions, the system can be modeled as synchronous using local timeouts to trigger evictions, thereby preserving both liveness and weak safety during periods of normal network latency. Near termination time, this same modeling proves useful.
    
    \item \textbf{Ingestion and Delivery Timeline}: The message travels over the point-to-point FIFO channel to the correct receiver $j$, triggering \textit{receive}. Since sender $i$ is correct, its local clock $ck_i$ increments monotonically without gaps, ensuring that $m.mid > \mathbf{LD}_j[i]$ (which starts at $-1$). The message passes the replay filter, the local causal history is updated ($\mathbf{LD}_j[i] \gets m.mid$), and the payload is successfully passed to the application layer via \textit{deliver}(m.pl). 
    
    \item \textbf{Cleanup Timeline}: Upon delivery, process $j$ immediately transmits a unicast $\text{ACK}$ bearing $m.mid$ back to process $i$ over the point-to-point network. Upon ingestion at process $i$, \textit{receive} locates the matching record in $\mathbf{U}_i[j]$ and purges it, completing the message lifecycle and reclaiming local space.
\end{enumerate}

A Byzantine process $b$ may attempt to disrupt this timeline and break system liveness using three distinct ways, which our algorithm successfully neutralizes:

\begin{itemize}
    \item \textbf{Withholding $\text{ACK}$s}: The adversarial objective is to cause memory exhaustion at the sender by refusing to acknowledge received messages, thereby locking resources. Our algorithm isolates this attack to the queue $\mathbf{U}_i[b]$. When $\mathbf{U}_i[b].\text{size} \ge M_{\max}$, a network eviction purges the head entry, freeing local memory and ensuring communication with honest peers continues unhindered.
    \item \textbf{Withholding $\text{PERMIT}$s}: The adversarial objective is to exploit the SPS invariant by withholding permission clearances, hoping to cause permanent causal blocking at honest senders. If an outbound queue $\mathbf{U}_i[dst]$ becomes full because its head message $m$ is blocked by a dependency on $b$ ($m.\textbf{DV}[b] > \mathbf{LPR}_i[b]$), the optimistic eviction architecture forcefully clears the stall by evicting and network-sending $m$. Sending downstream $\text{PERMIT}$s remains unaffected. 
    \item \textbf{Future Message Flooding}: The adversarial objective is to exhaust the receiver's memory by flooding it with massive, out-of-sequence message streams. Because the \textit{receive} procedure evaluates messages instantly against $\mathbf{LD}_i$ and never buffers out-of-order or future-dated payloads into localized holding structures, incoming garbage data is dropped or processed immediately without consuming queue capacity.
\end{itemize}

Across all adversarial scenarios, strictly enforced local space boundaries act as deterministic thresholds that trigger automatic, optimistic, or cascading evictions. These event-driven evictions forcefully resolve causal stalls and preserve processing fluidity within the application data plane. Combined with the reliable delivery guarantees of the underlying point-to-point channels between honest processes, the architecture mathematically ensures that no configuration of Byzantine processes can indefinitely freeze or exhaust the resources of an honest process, fully satisfying the liveness property for honest-to-honest communication.
\end{proof}

\subsection{Congestion-Relaxed Causal Delivery Guarantee}

\begin{theorem}[Congestion-Relaxed Causal Delivery]
Algorithm~\ref{alg:byzantine_sps_fifo_gc} guarantees Congestion-Relaxed Causal Delivery (Definition~\ref{def:crcd}) for all communications between honest processes in the system.
\end{theorem}

\begin{proof}
We prove each property of Definition~\ref{def:crcd}:
\begin{enumerate}
    \item \textbf{FIFO Order}: Since honest processes communicate over reliable point-to-point FIFO channels (e.g., TCP), and the receiver process $j$ processes and delivers messages from sender $i$ in sequence number order ($m.mid > \mathbf{LD}_j[i]$), the FIFO property is naturally preserved.
    
    \item \textbf{Bounded Causal Integrity}: %From Condition (2.a), we have that $i_1$ receives an ACK for $m_{i_0}$ from $i_{0}$ before $m_{i_0}$ is dequeued from $U_{i_1}[0]$ and from Conditon (2.b) at $i_1$, $i_1$ sends a normal (and not eviction-triggered) PERMIT to $i_2$, which implies $m_{i_0}$ has already been dequeued from $U_{i_1}[i_0]$. Thus $m_{i_0}$ has already been delivered to $i_0$ at the time $i_1$ sends a PERMIT to $i_2$ and before $i_2$ receives the PERMIT.
    Under the algorithm, when $i$ attempts to send PERMIT for $m$ via \textit{SendPermits()}, the PERMIT is blocked from transmission if $\exists x \neq i \land x \neq dst, m.\mathbf{LMS}[x] \ge \mathbf{U}_i[x][0].mid$. Thus, PERMIT($m$) can only be released to the network if either:
    \begin{itemize}
        \item The unblocking condition is satisfied, meaning $\forall x \neq i \land x \neq dst, m.\mathbf{LMS}[x] < \mathbf{U}_i[x][0].mid$ or $\mathbf{U}_i[x][0] = \perp$. This implies that $i$ has already received all required ACKs for messages it sent (to any $y$) before sending $m$ (or those messages were forcibly evicted from $\mathbf{U}_i[y]$ --- but that cannot happen for the case in Condition 2(a) of Definition~\ref{def:crcd} where $i=i_1$ and $y=i_0$). 
        \item Or, process $i$ experienced queue eviction in $\mathbf{Q}_i[x]$ and was forced to evict $m$, which will not happen for the case in Condition 2(b) of Definition~\ref{def:crcd} where $i=i_1$ and $x=i_2$.
    \end{itemize}

    Now from Condition 2(b), if process $i_1$ (the sender of $m_{i_1}$) has {\em causalSend} less than $Q_{max}$ messages to $i_2$ since inserting  $m_{i_1}$ in $Q_{i_1}[i_2]$ until dequeuing it from $Q_{i_1}[i_2]$, then $i_1$ sends a non-eviction driven PERMIT for $m_{i_1}$ to $i_2$. This implies from {\em SendPermits}() that $m_{i_1}.LMS[i_0] < U_{i_1}[i_0].mid$ or $U_{i_1}[i_0][0] = \perp$ which in turn implies that $m_{i_0}$ had been dequeued normally (or forcibly evicted) from $U_{i_1}[i_0]$. 

    However, Condition 2(a) implies that $m_{i_0}$ had been normally dequeued from $U_{i_1}[i_0]$ (on receiving an ACK from $i_0$) and not evicted. Therefore $m_{i_0}$ had been delivered to $i_0$ and its ACK received by $i_1$ before $i_1$ sends a PERMIT for $m_{i_1}$ to $i_2$. (Note, $i_1$ may have network-sent $m_{i_1}$ already and it may be delivered to $i_2$ already but that is okay.) 
    
    $i_2$ will not network-release $m_{i_2}$ to $i_3$ until receiving the PERMIT for $m_{i_1}$ from $i_1$, by Condition 2(c). 
    %--- if $m_{i_2}$ is not evicted from $\mathbf{U}_{i_2}[i_3]$. And $m_{i_2}$ will not be evicted by Condition 2(a). 
    Thus $m_{i_2}$ is network-sent by $i_2$ to $i_3$ after $m_{i_0}$ had been delivered to $i_0$.
    
    It therefore follows that, down the causal chain, $m_{i_{x}}, x\geq 2$ will not be delivered to $i_{x+1}$ before $m_{i_0}$ has been delivered to $i_0 (= i_{x+1})$. Bounded Causal Integrity follows. 

    The above logic assumed $m_{i_1}.per =$ True. If False, then $m_{i_0}$ was delivered to $i_0$ and its ACK received by $i_1$ before $i_1$ {\em causal-send} $m_{i_1}$ to $i_1$; from this Bounded Causal Integrity naturally follows.
    %
    %From Condition 2(b), each process $i_a, (2 \leq a \leq x-1)$ %network-sends $m_{i_a}$ 
    %sends a non-eviction-driven PERMIT for $m_{i_{a}}$ to $i_{a+1}$, implying it has already received an ACK for $m_{i_{a}}$ from $i_{a+1}$. Further, before it network-sent $m_{i_{a}}$, $i_a$ has received the PERMIT for $m_{i_{a-1}}$. Transitively, this implies $i_a$ network-sends $m_{i_a}$ after $m_{i_0}$ has been delivered to $i_0$.

    %Further, $i_x$ dequeues $m_{i_{x}}$ from $U_{i_x}[i_{x+1}]$ normally implying it has already received an ACK for it from $i_{x+1}$, but at any rate $m_{i_x}$ has been network-sent (and hence delivered to $i_{x+1} = i_0$) after $m_{i_0}$ has been delivered to $i_0$.  

    %Thus if no process evicts more than the stipulated number of messages, the safety condition $\forall x \neq i \land x \neq dst, m_2.\mathbf{DV}[x] \le \mathbf{LPR}_i[x]$ holds, ensuring $i_0$ never delivers $m__{i_x}$ before $m_{i_0}$. Hence, Bounded Causal Integrity is satisfied.
\end{enumerate}
The theorem follows.
\end{proof}

\subsection{Safety Trade-offs}

\begin{table*}%[t]
\centering
\begin{center}
%{\small
\begin{minipage}{\textwidth}
\begin{center}
\begin{tabular}{|l|l|l|}
%\toprule %
\hline
 & Attack by Byzantine $b$ & Analysis of impact on correct processes. \\
%\midrule %
\hline 
\hline

1 & $b$ witholds ACK to $i$ & eventually $\textbf{U}_i[j]$ fills up; $m$ dequeued \& sent, added to $\textbf{Q}_i[j]$;  \\
  & $i$ has queued an $m$ in $\textbf{U}_i[j]$ & eventually $\textbf{Q}_i[j]$ fills up, $m$ evicted and PERMIT sent to $j$. {\bf No impact}.\\
\hline

2 & $b$ sends early ACK for $m1$ to $i$,  & $j$ sends early sends $m2$ to $b$, $b$ may receive $m2$ before $m1$ \\
 &  this causes $i$ to send early $m$ and PERMIT to $j$ & {\bf No impact.} Strong safety violation possible only at $b$. \\
 &  which $j$ receives, $j$ queues $m2$ in $\textbf{U}_i[b]$  & \\
\hline
%2 & $b$ sends early ACK to $i$ & $i$ sends early PERMIT to $j$. {\bf Strong safety violation possible at $b$}.\\
%  & $i$ has queued an $m$ in $U_i[j]$ & \\
%\hline

3 & $b$ witholds PERMIT to $i$ & eventually $\textbf{U}_i[j]$ fills up; $m$ dequeued \& sent, added to $\textbf{Q}_i[j]$; \\
  & $i$ has queued an $m$ in $\textbf{U}_i[j]$ & eventually $\textbf{Q}_i[j]$ fills up, $m$ evicted and PERMIT sent to $j$. {\bf No impact}. \\
\hline

4 & $b$ sends $m_{b0}$ to $j$, then $m_{b1}$ to $i$, & $i$ sends early $m$ to $j$. $j$ receives $m$ before $m_{b0}$. \\
  & $b$ sends early PERMIT for $m_{b1}$ to $i$,  & {\bf Strong safety violation (messages from $b$ delivered later) possible.}\\
 & $i$ receives $m_{b1}$ and PERMIT, queues an $m$ in $\textbf{U}_i[j]$  & \\
\hline

5 & $b$ sends MSG with low $mid$ to $j$ & $j$ deletes on reception. {\bf No impact.}\\
\hline

6 & $b$ sends MSG with high $mid$ to $j$ & may affect its own FIFO channel to $j$. {\bf No impact}.\\
\hline

7 & $b$ sends ACK with low $mid$ to $i$ & Equiv. to witholding ACK (Attack 1). {\bf No impact.}  \\
\hline

8 & $b$ sends ACK with high $mid$ to $i$ & Equiv. to sending early ACK (Attack 2). \\
  & $\textbf{U}_i[b][0].mid < mid \leq \textbf{U}_i[b].last.mid$ & {\bf No impact.} Strong safety violation possible only at $b$. \\
\hline

9 & $b$ sends ACK with high $mid$ to $i$ & Equiv. to witholding ACK (Attack 1). {\bf No impact.}\\
  & $mid > \textbf{U}_i[b].last.mid$ & \\
\hline 

10 & $b$ sends PERMIT with low $mid$ to $i$ & Ignored by $i$. Equiv. to witholding PERMIT (Attack 3). {\bf No impact.}  \\
\hline

11 & $b$ sends PERMIT with high $mid$ to $i$ & Equiv. to sending early PERMIT (Attack 4). \\
  & $\textbf{LD}_i[b] \geq mid > \textbf{LPR}_i[b]$ & {\bf Strong safety violation (message from $b$ delivered later) possible.} \\
\hline

12 & $b$ sends PERMIT with high $mid$ to $i$ & Equiv. to witholding PERMIT (Attack 3). {\bf No impact.}\\
  & $mid > \textbf{LD}_i[b]$ & \\
\hline 

13 & $i$ sends $m_{i0}$ to $j$, then $m_{i1}$ to $b$, $b$ receives $m_{i1}$; & $j$ may receive $m_{b}$ before $m_{i0}$. \\
  & (a) $b$ updates $m'.pl$ for unsent $m'$ in $U_b[*]$ based on   & {\bf Strong safety violation (messages from $b$ delivered earlier) possible.}\\
 & $m_{i1}.pl$ \& sends the $m'$ without their required PERMITs;  & Also, causalSend concurrently (with $m_{i1}$) $m'$ may see $m_{i1}.pl$ at dests. $*$.  \\
 & (b) $b$ without waiting for PERMIT for $m_{i1}$ sends $m_{b}$ to $j$ & \\
\hline

%\bottomrule 
\end{tabular}
\end{center}
\end{minipage}
%}
\end{center}
%\end{minipage}
%\medskip
\caption{Attacks by a Byzantine process $b$ on correct processes and an analysis of their impact. $i,j$ are correct processes.
}
\label{table:attacks}
\end{table*}

Traditional causal delivery algorithms require \textit{Strong Safety}, which dictates that for any message $m' \in CP(m)$ destined for the same correct process, the receiver must never deliver $m$ before $m'$. %\cite{DBLP:journals/tpds/MisraK24}. 
However, %as proven by Theorem 8 of Misra and Kshemkalyani, 
it is impossible to deterministically guarantee both strong safety and liveness for unicasts without using cryptography in an asynchronous Byzantine-faulty system \cite{DBLP:journals/tpds/MisraK24}. Because an adversary can weaponize strict safety constraints to permanently halt honest nodes, our algorithm purposefully relaxes these ordering bounds under physical capacity pressure to preserve global system liveness. 

%The different possible safety attacks are shown in 
As classified in 
Table~\ref{table:attacks}, the behavioral consequences of the safety relaxations fall into three distinct analytical categories based on the structural boundaries of the causal past.

\begin{enumerate}
    \item \textbf{Self-Contained Violations (Attacks 2, 8)}: The Byzantine process reports a premature acknowledgment for a message it has not yet processed. While this allows the honest sender to advance its local queue state, any subsequent out-of-order delivery is contained entirely within the internal execution environment of the Byzantine process $b$. Because an adversary can inherently choose to arbitrarily permute or reorder its private application-level state transitions, this scenario does not constitute a true breach of network-wide causal integrity.
    
    \item \textbf{Adversarial Sequence Shifting (Attacks 4, 11)}: In this scenario, a permit is forged or prematurely cleared for an application message $m_{b1}$ whose causal predecessor is $m_{b0}$ ($send_b(m_{b0}) \rightarrow send_b(m_{b1})$). When an honest process $i$ delivers $m_{b1}$ and subsequently invokes \textit{causalSend} to transmit message $m$ to a common destination $j$, the standard happens-before relation dictates that $m_{b0} \rightarrow m_{b1} \rightarrow m$, meaning $m_{b0} \in CP(m)$. If process $j$ delivers $m$ before $m_{b0}$ arrives, this technically constitutes a violation of \textit{Strong Safety}. However, as the attacking Byzantine process's message $m_{b0}$ arrives later, it cannot gain any gaming application advantage at $j$, so this attack cannot be considered a harmful breach of network-wide causal integrity.
    
    Morever, this sequence shifting \textit{does not violate Weak Safety}. Because the causal link $m_{b0} \rightarrow m_{b1}$ occurs as a local state transition within the Byzantine process $b$, it fails the structural preconditions of the Byzantine happens-before relation ($\stackrel{B}{\rightarrow}$), which strictly requires all processes executing transitions along the path to be correct. Consequently, $m_{b0} \not\stackrel{B}{\rightarrow} m_{b1}$, which excludes the Byzantine message from the Byzantine causal past of the honest message ($m_{b0} \notin BCP(m)$). %This demonstrates that the causal consistency of honest nodes remains perfectly intact under $\stackrel{B}{\rightarrow}$, as the broken dependency chain is localized entirely within the adversary's manipulated workspace.
    
    \item \textbf{True Safety Breaches (Attack 13)}: This represents a structural safety violation where an active dependency chain is explicitly broken over the wire, such that $send(m_{i0}) \rightarrow send(m_b)$ but the receiver $j$ delivers the adversarial payload $m_b$ before the honest payload $m_{i0}$. Because network-layer cryptography cannot validate or inspect semantic dependencies that a Byzantine actor intentionally erases within its private application workspace prior to transmission, this violation is mathematically unpreventable in an asynchronous setting using any deterministic solution and must be accepted to maintain liveness.
\end{enumerate}

\subsubsection{Weak Safety Violation under Severe Asynchronous Network Delays}
While adversarial manipulation can only induce Strong Safety anomalies on Byzantine-derived paths, physical buffer limitations can occasionally force a true violation of \textit{Weak Safety} among completely honest processes under extreme asynchronous network conditions.

Consider a fully non-adversarial execution trace involving three correct processes: $i$, $j$, and $k$. Process $i$ transmits an application message $m_1$ to process $k$, and then queues a causally dependent message $m_2$ for process $j$, establishing the relation $send_i(m_1) \rightarrow send_i(m_2)$. Because all participating nodes are correct, this sequence strictly satisfies the parameters of the Byzantine happens-before relation, yielding $m_1 \stackrel{B}{\rightarrow} m_2$. Process $j$ subsequently delivers $m_2$ and sends a new message $m''$ to process $k$, which transitively embeds $m_1$ into its Byzantine causal past.

If the communication link between $i$ and $k$ experiences severe physical network congestion or extreme asynchronous scheduling latency, the acknowledgment ($\text{ACK}$) for $m_1$ will be indefinitely delayed. From the isolated perspective of the honest sender $i$, this persistent silence is structurally indistinguishable from a malicious acknowledgment withholding attack (Attack 1). As process $i$ continues to process subsequent application demands, its isolated queues will inevitably saturate, reaching the absolute capacity thresholds ($M_{\max}, Q_{\max}$). 

To guarantee system-wide progress and prevent a permanent local deadlock, process $i$ is forced to execute an optimistic network eviction. This action forcefully removes the unacknowledged tracking entries and dispatches a premature $\text{PERMIT}$ over the point-to-point network to unblock the channel to $j$. If process $j$ ingests this $\text{PERMIT}$ and immediately transmits $m''$ to process $k$, process $k$ may receive and deliver $m''$ before the severely delayed message $m_1$ physically arrives at its network interface. 

Because $m_1 \in BCP(m'')$ and both messages are sent to the same correct process $k$, delivering $m''$ before $m_1$ constitutes a weak safety violation. However, such a violation occurs with very low probability because the SPS dependencies of the stalled permits (which are $M_{\max} + Q_{\max}$ messages old or older on the same channel from correct $i$ to correct $j$) would have been satisfied by then due to the long-tail properties of the network latency distribution. So the protocol accepts this rare weak safety anomaly as a necessary mathematical trade-off; prioritizing absolute ordering safety over severe asynchronous delay would require unbounded local queue structures ($M_{\max} \to \infty$), which would instantly reintroduce catastrophic vulnerabilities to Byzantine memory exhaustion and distributed denial-of-service flooding. In any case, {\em Congestion-Relaxed Causal Delivery} is guaranteed.

\section{Complexity Analysis}
\label{sec:complexity}

Our process-isolated optimistic causal delivery protocol is engineered to combine resilient Byzantine fault containment with highly efficient network resource consumption. By completely avoiding multi-round cryptographic consensus primitives and heavy global broadcasts, the architecture maintains an efficient pure unicast footprint. We evaluate the algorithm's complexity profile across four dimensions: network message metadata overhead, communication message complexity, local memory space complexity, and computational execution time, as shown in Table~\ref{table:comparison}.
A detailed analysis is given in Appendix~\ref{ap:complexity}.

\section{Discussion and Conclusions}
\label{sec:conc}
In this paper, we presented the first point-to-point unicast causal delivery protocol that tolerates Byzantine failures in a purely asynchronous environment while preserving a strict $O(1)$ network message space overhead. By introducing the {\em Congestion-Relaxed Causal Delivery} abstraction, we established a robust mathematical bridge between absolute liveness and relaxed causal safety bounds.

Our protocol offers several distinct advantages over existing approaches:
\begin{itemize}
    \item \textbf{Optimal Network Efficiency}: Unlike the existing algorithms that append $O(n)$ or $O(n^2)$ dependency vectors to message payloads and have $O(n^2)$ or $O(n^3)$ message complexity, our protocol preserves a strict $O(1)$ wire overhead, maximizing network throughput in large-scale systems.
    \item \textbf{Decoupling and Fault Isolation}: By routing all messages over point-to-point unicast channels and managing state via isolated outbound queues ($\mathbf{U}_i$) and missing permit queues ($\mathbf{Q}_i$), we completely isolate process faults. A malfunctioning or Byzantine node cannot exhaust the resources or freeze the channels of unrelated correct processes.
    \item \textbf{High Scalability}: The event-driven, localized nature of our signaling plane results in an amortized computational execution cost of just $O(n)$ per message, avoiding heavy global consensus overhead during normal execution.
\end{itemize}

A key conceptual takeaway of our work is the inevitability of safety trade-offs in Byzantine settings. As shown by Misra and Kshemkalyani~\cite{DBLP:journals/tpds/MisraK24}, guaranteeing both Strong Safety and liveness deterministically without using cryptography is impossible in asynchronous Byzantine networks. Even if one were to attempt enforcing the  \textit{Weak Safety} guarantee, detecting it in practice is fundamentally unfeasible: verifying Weak Safety requires identifying which processes are correct and which are Byzantine, a task that is itself impossible to achieve in an asynchronous system. 
Consequently, our newly introduced {\em Congestion-Relaxed Causal Delivery} abstraction is not only theoretically sound but practically essential. It provides a formal, observable safety guarantee that is useful for application developers: causal ordering is guaranteed for all honest-to-honest communication paths unless the specified bounds are exceeded. Crucially, this threshold is completely under the control of the application designer, who can set the queue capacities $M_{\max}$ and $Q_{\max}$ to match their system's memory constraints and congestion tolerances. By raising these boundaries, designers can arbitrarily reduce the probability of weak safety violations at a low, predictable local memory cost of $O(n^2)$. 
%Despite fine-tuning queue capacities, if $\mathbf{U}_i$ ($\mathbf{Q}_i$) often have eviction because expected PERMITs (ACKs) from a/some certain $p_j$ never arrive preventing normal network-sending of MSGs (normal sending of PERMITs), that $p_j$ is suspected to be Byzantine; a rational $p_j$ would therefore not behave in that manner.
Despite fine-tuning queue capacities, if $\mathbf{U}_i[k]$ often have eviction because expected PERMITs (ACKs) from a/some certain $p_j$ (from $p_k$) never arrive preventing normal network-sending of MSGs in $\mathbf{U}_i[k]$ (normal dequeuing from $\mathbf{U}_i[k]$), those $p_j$ ($p_k$) are suspected to be Byzantine; a rational $p_j$ ($p_k$) would therefore not behave in that manner.
Likewise, if $\mathbf{Q}_i$ often have eviction because expected ACKs from a/some certain $p_j$ never arrive preventing normal sending of PERMITs, those $p_j$ are suspected to be Byzantine; a rational $p_j$ would therefore not behave in that manner.

Ultimately, this work demonstrates that the trade-off between buffer capacities and transient ordering relaxations represents a highly viable and optimized paradigm for secure, high-performance distributed systems.

% \bibliographystyle{plain}
%\bibliography{references} 

%\input{ICDCN_body_Revised}
%\input{DISCbody_Revised}

\sloppy

\balance
%\vspace*{-8pt}
\bibliographystyle{ACM-Reference-Format}
\bibliography{references} 

\appendix
\section{Usage of Data Structures}
\label{ap:ds}

\subsubsection{Local Dependency Vector ($DV$)}
The Dependency Vector ($\mathbf{DV}$) is a $1$-dimensional array of size $n$ with natural numbers embedded within the internal structure of each message record $m$. Its primary function is to maintain a local, point-in-time snapshot of the highest message sequence numbers delivered to process $i$ from every peer at the moment of message initialization. Crucially, the $\mathbf{DV}$ remains completely local to the executing process as it is stored within the unacknowledged application queues $\mathbf{U}_i$ but is explicitly omitted from the transmitted network payload. 

By preventing the transmission of variable-sized metadata over the wire, the protocol preserves a constant $O(1)$ network message size overhead. However, because each of the $n$ isolated buffers can hold up to $M_{\max}$ records—with each record caching an $n$-sized array—the local space complexity scales to $O(n^2)$. Within the \textit{sendMessages} procedure, the $\mathbf{DV}$ enforces causal ordering by blocking the transmission of any message whose dependencies have not been cleared by the current permit state, requiring that $\forall j \neq i \land j \neq dst, m.\mathbf{DV}[j] \le \mathbf{LPR}_i[j]$.

\subsubsection{Last Message Sent Vector $(LMS_i)$}
The Last Message Sent vector ($\mathbf{LMS}_i$) is a $1$-dimensional tracking array of size $n$ used to log the local sequence number ($mid$) of the most recent message transmitted from process $i$ to each respective peer. This vector is instrumental in preserving the constraints of the Sender Permission to Send (SPS) invariant on the sender side. 

During the execution of $\textit{causalSend}$, a localized snapshot of $\mathbf{LMS}_i$ is copied directly into the message record as $m.\mathbf{LMS}$, capturing the exact outbound causal history across all channels at the moment of creation. Following this snapshot, $\mathbf{LMS}_i[dst]$ is advanced to $m.mid$. The core utility of this snapshot occurs within the \textit{sendPermits} procedure, where $m.\mathbf{LMS}$ is evaluated against the current heads of the outbound queues $\mathbf{U}_i$ to ensure that all causal predecessors sent to other destinations prior to $m$ have stabilized/been delivered before a permit is released.

%The core utility of this vector occurs within the $\textit{receive(ACK)}$ handler. When an $\text{ACK}$ arrives for message $m$, the algorithm checks $m.\mathbf{LMS}$ to verify whether all causal predecessors sent to other destinations prior to $m$ have stabilized. This safety condition is confirmed by ensuring that for all channels $j$, the value of $m.\mathbf{LMS}[j]$ is strictly less than the sequence number of the message currently at the head of queue $\mathbf{U}_i[j]$ (or that the queue is empty), preventing premature permit generation.

\subsubsection{Isolated Application Queues ($\mathbf{U}_i$)}
Outbound application messages destined for peer processes are segregated into an array of isolated FIFO queues, $\mathbf{U}_i[1..n]$, each strictly bounded by a maximum capacity constraint $M_{\max}$. This architectural isolation establishes robust fault containment. If a Byzantine process maliciously suppresses $\text{ACK}$ control signals, memory consumption is strictly confined to the buffer $\mathbf{U}_i[dst]$ dedicated to that specific channel, ensuring parallel communication channels with honest peers remain completely unaffected. 

When an isolated queue reaches its capacity boundary, the protocol executes an optimistic network eviction, removing the oldest stale message from the head of the queue to force local progress. If the evicted message %was previously transmitted and carried a permit requirement ($stale\_m.sent \wedge stale\_m.per$), 
carried a permit requirement ($stale\_m.per$), it undergoes a cascading capacity check against the signaling plane buffers ($\mathbf{Q}_i$) to ensure memory is reclaimed before the state is advanced.

\subsubsection{Isolated Missing Permits Queues ($\mathbf{Q}_i$)}
The Missing Permits Queue array, $\mathbf{Q}_i[1..n]$, consists of $n$ independent FIFO buffers bounded by a maximum capacity constraint $Q_{\max}$. This data structure explicitly decouples acknowledgment processing from permit generation. When an inbound $\text{ACK}$ confirms that a message requiring a permit has been delivered, or when such a message is optimistically evicted from $\mathbf{U}_i$, the record is transferred into $\mathbf{Q}_i[src]$. 

The message record remains securely buffered within this isolated memory space until the sender's local outbound history satisfies the stability invariants evaluated by \textit{sendPermits}. If $\mathbf{Q}_i[src]$ reaches its capacity limit $Q_{\max}$, the algorithm executes a defensive cascading eviction, forcefully removing the oldest pending record from the head of $\mathbf{Q}_i[src]$ and dispatching a proactive $\text{PERMIT}$ over the wire to prevent local memory exhaustion.

\begin{comment}
\subsubsection{Eviction Tracker ($ET_i$)}
The Eviction Tracker ($\mathbf{ET}_i$) is a $1$-dimensional array of size $n$ of natural numbers that records the high-water mark of message sequence numbers that have been optimistically evicted for each destination channel. During a network eviction event in the \textit{causalSend} procedure, $\mathbf{ET}_i[dst]$ is updated to the sequence number ($mid$) of the discarded head message. %using a $\max$ function to cache the sequence number ($mid$) of the discarded head message. 
This structure serves as a critical safety filter within the signaling plane: if a delayed or retransmitted $\text{ACK}$ subsequently arrives from a slow or adversarial process, its sequence number is validated against $\mathbf{ET}_i[src]$. Any acknowledgment satisfying $m\_id \le \mathbf{ET}_i[src]$ is immediately discarded, preventing the system from processing control information for historical states that have already been optimistically cleared. 
\end{comment}

\subsubsection{Last Delivered Message IDs ($LD_i$)}
The Last Delivered vector ($\mathbf{LD}_i$) is a $1$-dimensional array of size $n$ tracking the sequence number of the most recent valid application message delivered from each peer process in the network. Serving as the node's local causal ground truth, $\mathbf{LD}_i$ is utilized during the execution of $\textit{causalSend}$ to populate the $\mathbf{DV}$ fields of newly initialized outbound payloads. Additionally, $\mathbf{LD}_i$ provides immediate protection against adversarial replay attacks within the message ingestion pipeline. Upon receiving an incoming payload, the \textit{receive} procedure verifies that $m.mid > \mathbf{LD}_i[src]$. Any message failing this check is classified as a duplicate or an active forgery and is immediately dropped, preserving the integrity of the delivery layer.

\subsubsection{Last Permit Received Vector $(LPR_i)$}
The Last Permit Received vector ($\mathbf{LPR}_i$) is an array of sequence numbers tracking the highest causal authorization received from each peer process. This vector acts as a localized causal ceiling, ensuring the preservation of the SPS invariant without requiring a globally synchronized consensus plane. A process monotonically advances its $\mathbf{LPR}_i[src]$ value when an incoming message arrives with a disabled permit flag ($m.per = \text{False}$), indicating that the payload carries no unresolved upstream dependencies, or upon the explicit receipt of a dedicated $\text{PERMIT}$ control packet. 
$\mathbf{LPR}_i[src]$ value is also advanced even if there are unresolved upstream dependencies that $p_i$ is forced to ignore because of overflow of some $\mathbf{U}_i[x]$ queue.
This vector is continuously evaluated by \textit{sendMessages} to determine whether buffered outbound payloads can be safely released to the network.

\section{Description of Algorithm Procedures}
\label{ap:procs}

\subsubsection{\textit{causalSend(dst, p)}} 
The \textit{causalSend} procedure is executed when the application layer dispatches a payload $p$ to a designated target process $dst$. This procedure handles capacity validation, network eviction, and the local logging of causal history vectors prior to transmission.

\begin{enumerate}
    \item \textbf{Buffer Capacity and Network Eviction}: Before constructing the outbound message, the algorithm verifies the current memory footprint of the targeted application buffer, $\mathbf{U}_i[dst]$. If the buffer size meets or exceeds the maximum capacity threshold $M_{\max}$, a network eviction is executed. The oldest message ($stale\_m$) is systematically removed from the head of $\mathbf{U}_i[dst]$. %, and the local eviction tracker is updated via $\mathbf{ET}_i[dst] \gets \max(\mathbf{ET}_i[dst], stale\_m.mid)$ to ensure late-arriving control signals for this entry are safely ignored.
    \item \textbf{Forced Transmission and Permit Cascading}: If the evicted $stale\_m$ has not yet been transmitted over the network ($\neg stale\_m.sent$), it is flagged as sent ($stale\_m.sent \gets \text{True}$) and immediately dispatched over the point-to-point network via $\text{send}_i(dst,\- \text{MSG}, stale\_m)$. 
    The $\mathbf{LPR}_i$ vector is also updated using the dependency vector of $stale\_m$ ($\forall x, \mathbf{LPR}_i[x] \gets \max(\mathbf{LPR}_i[x], stale\_m.\mathbf{DV}[x])$) even if some of those dependencies were not satisfied but had to be ignored because of queue overflow. This allows blocked messages in other queues to be network-sent.
    %If $stale\_m$ was already transmitted but remains unacknowledged while holding an active permit dependency ($stale\_m.sent \wedge stale\_m.per$), a 
    If $stale\_m$ holds an active permit dependency $(stale\_m.per)$, a 
    cascading capacity check is performed on the corresponding missing permits queue $\mathbf{Q}_i[dst]$. If this queue is saturated ($\ge Q_{\max}$), its head element $m'$ is removed and a proactive $\text{PERMIT}$ message carrying $m'.mid$ is transmitted over the point-to-point network to unblock the destination. The evicted $stale\_m$ is then appended to $\mathbf{Q}_i[dst]$. Following this sequence, $\textit{sendPermits}()$ is invoked to evaluate and flush pending permits.
    \item \textbf{Message Initialization and Snapshotting}: Once buffer space is secured, the algorithm determines the permit requirement flag ($req\_per$), setting it to $\text{True}$ if any alternative application queue contains pending unacknowledged traffic ($\exists x \neq dst : \mathbf{U}_i[x].\text{size} > 0$). A new message record $m$ is initialized with the target destination, the current sequence number $ck_i$, the calculated permit flag, and the payload. The local clock $ck_i$ is then incremented. Crucially, the message captures a point-in-time snapshot of the node's local causal history by copying $\mathbf{LD}_i$ into its local Dependency Vector ($m.DV$) %($\mathbf{U}_i[dst].\text{last}.\mathbf{DV}_i$) 
    and logging the outbound channel states $\mathbf{LMS}_i$ into $m.\mathbf{LMS}$. %$\mathbf{U}_i[dst].\text{last}.LMS_i$. 
    \item \textbf{Pipeline Execution}: Finally, the local tracking entry is updated ($\mathbf{LMS}_i[dst] \gets m.mid$), and $\textit{sendMessages}()$ is executed to attempt an immediate, optimistic transmission of the newly buffered payload.
\end{enumerate}

\subsubsection{\textit{sendMessages()}}
The \textit{sendMessages} procedure is the primary component governing message transmission, responsible for validating that outbound payloads satisfy the SPS invariant before they are sent over the point-to-point network.

The procedure loops systematically across all process-isolated application channels from $dst = 1$ to $n$. For each isolated buffer $\mathbf{U}_i[dst]$, the algorithm processes messages sequentially starting from index $k=0$. For every unsent record ($\neg m.sent$), it evaluates whether the payload is causally constrained by verifying its embedded dependency history against the local permit vector:
$ \mathbf{blocked} \gets \neg \left( \forall j \neq i \land j \neq dst, \left( m.\mathbf{DV}[j] \le \mathbf{LPR}_i[j] \right) \right) $.

If the message is unblocked ($\neg \mathbf{blocked}$), the algorithm sets $m.sent \gets \text{True}$ and transmits the packet as a unicast $\text{MSG}$ message over the point-to-point network. Note however, that the $m.\mathbf{DV}$ and $m.\mathbf{LMS}$ fields of $m$ should not be included in $\text{MSG}$. If the safety check fails, indicating that the message depends on a causal predecessor that has not yet been authorized or cleared by an incoming permit, the evaluation loop for that specific channel immediately executes a $\textbf{break}$ statement. Because strict point-to-point FIFO ordering must be preserved within each channel, a single blocked message at the head of a queue halts the transmission of all subsequent records in that isolated buffer, preventing causal fractures without affecting unrelated destination channels.

\subsubsection{\textit{sendPermits}()}
The \textit{sendPermits} procedure manages the systematic release of pending causal \textit{PERMIT}s stored within the isolated permit queues back to peer processes via point-to-point unicast messaging.

The procedure loops sequentially through all destination channels $dst \in [1..n]$. For each channel, an evaluation loop runs as long as the respective permit queue is non-empty ($\mathbf{Q}_i[dst] \neq \emptyset$). The algorithm peeks at the message record $m$ situated at the head of the queue ($\mathbf{Q}_i[dst].\text{head}$) and subjects it to a strict outbound causal stability check:
$ \forall j \neq i \land j \neq dst, \quad \left( m.\mathbf{LMS}[j] < \mathbf{U}_i[j][0].mid \quad \vee \quad \mathbf{U}_i[j][0] = \bot \right) $.

This condition mathematically verifies that all historical application messages transmitted by process $i$ to any other peer $j$ prior to the generation of $m$ have either been explicitly acknowledged by their recipients or cleared via network eviction. If this stability check is satisfied, it is safe to unblock the target process; a unicast $\text{PERMIT}$ message bearing $m.mid$ is sent to $dst$ over the point-to-point network, and the record is permanently popped from $\mathbf{Q}_i[dst]$. If the stability condition is not met, the procedure triggers a $\textbf{break}$ for that channel, preserving strict FIFO pipeline behavior within the signaling queue.

\subsubsection{\textit{receive(src, MSG, m)}}
This procedure handles the real-time ingestion of point-to-point application messages ($\text{MSG}$) arriving from a source node $src$. Given the Byzantine network context, incoming payloads must pass through sequencing filters before application-level delivery occurs.

\begin{enumerate}
    \item \textbf{Replay and Duplicate Protection}: Upon message arrival, the algorithm executes an immediate sequence check against the local ground truth vector, evaluating whether $m.mid \le \mathbf{LD}_i[src]$. If true, the message is identified as a stale duplicate or an active Byzantine replay attack; the payload is immediately dropped, and the procedure terminates, preserving at-most-once delivery constraints.
    \item \textbf{Causal Ceiling Optimization}: If the incoming message is verified as valid and carries a disabled permit requirement flag ($\neg m.per$), it implies that the sender has no outstanding unpermitted causal dependencies on other channels. The receiver immediately optimizes its control state by advancing its authorization vector: $\mathbf{LPR}_i[src] \gets m.mid$. This update instantly elevates the causal ceiling for outgoing messages awaiting clearance from $src$.
    \item \textbf{State Synchronization and Acknowledgment}: To conclude the ingestion sequence, the receiver records the progress by setting $\mathbf{LD}_i[src] \gets m.mid$. A unicast control acknowledgment ($\text{ACK}$) containing the sequence index $m.mid$ is then returned to the source node over the point-to-point network to permit buffer reclamation at the sender. Following state synchronization, the raw payload is delivered to the local application layer via \textit{deliver}.
\end{enumerate}

\subsubsection{\textit{receive(src, ACK, m\_id)}}
This procedure is executed upon receiving an acknowledgment control message ($\text{ACK}$) from a process $src$ confirming the delivery of a previously transmitted message $m\_id$.

The algorithm first deploys dual defense filters to protect buffer states from malformed or out-of-order signaling flags:
\begin{enumerate}
    \item \textbf{Eviction or Replay Filtering}: If $m\_id < \mathbf{U}_i[src][0].mid$, the incoming signal corresponds to a historical record that has already been cleared from memory via an optimistic network eviction or after already receiving an $\text{ACK}$. This received (duplicate) $\text{ACK}$ is discarded immediately.
    \item \textbf{Sequence Inflation Defense}: The procedure verifies if $m\_id > \mathbf{U}_i[src].\text{last}.mid$. If true, the incoming acknowledgment references a message sequence number that was never generated or transmitted by process $i$, indicating an active Byzantine sequence inflation attack. The control signal is dropped.
\end{enumerate}

If the acknowledgment passes both filters, the algorithm processes the isolated application queue $\mathbf{U}_i[src]$ while it is non-empty. As long as the oldest message at index 0 satisfies $\mathbf{U}_i[src][0].mid \le m\_id$, it is sequentially removed from the queue. For each removed message $m$, if it was transmitted and required a permit ($m.sent \wedge m.per$), the algorithm manages its transition to the permit tracking structure. If $\mathbf{Q}_i[src].\text{size} \ge Q_{\max}$, a defensive eviction pops the oldest entry $m'$ from the head of $\mathbf{Q}_i[src]$ and transmits a proactive permit carrying $m'.mid$ over the point-to-point network. The message $m$ is then appended to the tail of $\mathbf{Q}_i[src]$. If an entry does not satisfy the loop condition, the loop breaks. Finally, the procedure calls \textit{sendPermits}() to evaluate and release any unblocked permits.

\subsubsection{\textit{receive(src, PERMIT, m\_id)}}
This procedure processes incoming authorization control messages ($\text{PERMIT}$) dispatched by a peer process $src$ to clear outstanding causal barriers.

To prevent adversarial nodes from manipulating local execution timelines or inducing safe-state regressions, the permit index is subjected to a strict bounding verification check:
$\mathbf{LD}_i[src] \ge m\_id > \mathbf{LPR}_i[src]$.

The upper boundary ($\mathbf{LD}_i[src] \ge m\_id$) guarantees that a permit is only processed if it corresponds to an application message that has already been received and delivered locally by process $i$. This prevents Byzantine processes from forcing state transitions based on future-dated or fabricated control data. The lower boundary ($m\_id > \mathbf{LPR}_i[src]$) ensures strict structural monotonicity, causing the algorithm to ignore stale or redundant permits. 

Once validated, the local permit tracking vector is advanced by setting $\mathbf{LPR}_i[src] \gets m\_id$. The procedure then immediately invokes \textit{sendMessages}. This re-evaluates the dependency status of all outbound application messages currently stalled across the isolated queues $\mathbf{U}_i$, allowing the process to immediately flush newly unblocked payloads onto the data plane.

\section{Detailed Complexity Analysis}
\label{ap:complexity}

\subsubsection{Message Space Overhead}

A primary design objective of the protocol is the minimization of causal metadata transmitted over the network wire. Every application packet ($\text{MSG}$) transmitted across point-to-point FIFO channels carries only a fixed, constant number of scalar fields, specifically: the target destination identifier ($rcv$), the local sequence index ($mid$), a single boolean permit requirement flag ($per$), and the application payload ($pl$). 

Crucially, as specified in the structural definitions of the message record, Local Dependency Vector ($m.\mathbf{DV}$) and the outbound causal history snapshot ($m.\mathbf{LMS}$) are maintained strictly within local memory buffers. Because these vectors are explicitly omitted from the serialized network payload, the message space overhead on the wire remains strictly $O(1)$ with respect to the network size $n$.

\subsubsection{Message Complexity (Number of Messages)}

The protocol operates exclusively on a single point-to-point data and signaling plane, executing all state updates over unicast communication channels. The message complexity generated per application-level transmission invocation is bounded as follows:
\begin{itemize}
    \item Exactly one application message packet ($\text{MSG}$) is unicast from the sender to the target destination. Size is $O(1)$.
    \item Exactly one control acknowledgment ($\text{ACK}$) is returned from the receiver to the sender to enable local buffer clearance. Size is $O(1)$.
    \item At most one explicit control permission message ($\text{PERMIT}$) is generated and unicast to the receiver to authorize downstream deliveries and satisfy the constraints of the Sender Permission to Send (SPS) invariant. Size is $O(1)$.
\end{itemize}
Because all tracking, verification, and clearance operations are handled strictly via localized, point-to-point unicast messaging, the message complexity per application message send is strictly $O(1)$.

It also follows that the communication word complexity is $O(1)$.

\subsubsection{Local Space Overhead}
To achieve perfect fault isolation and prevent adversarial actors from triggering head-of-line blocking or resource exhaustion on honest channels, the protocol organizes its tracking parameters into process-segregated data structures, resulting in an $O(n^2)$ local memory footprint:
\begin{itemize}
    \item \textbf{Global Tracking Arrays}: Each node maintains flat, tracking vectors for basic accounting, including the %Eviction Tracker ($\mathbf{ET}_i$), 
    Last Delivered message IDs ($\mathbf{LD}_i$), the Last Permit Received ceiling ($\mathbf{LPR}_i$), and the local Last Message Sent pointers ($\mathbf{LMS}_i$). Combined, these structures consume $3 \times O(n) = O(n)$ space.
    \item \textbf{Isolated Queue Infrastructures}: Outbound traffic and pending permissions are managed across an array of $n$ independent channels for the unacknowledged message buffers ($\mathbf{U}_i[1..n]$) and the missing permit buffers ($\mathbf{Q}_i[1..n]$). These queues are strictly bounded by the physical capacity constants $M_{\max}$ and $Q_{\max}$, respectively.
    \item \textbf{Embedded Vector Overhead}: Every individual message record cached within the $\mathbf{U}_i$ or $\mathbf{Q}_i$ arrays embeds two localized, $n$-sized vectors ($m.\mathbf{DV}$ and $m.\mathbf{LMS}$). 
\end{itemize}
Consequently, with $n$ distinct destination channels maintaining a combined capacity bounded by $(M_{\max} + Q_{\max})$ records, and each record caching an $n$-sized tracking state, the total local space complexity scales as $O(n \times (M_{\max} + Q_{\max}) \times n) = O(n^2)$. This quadratic memory commitment represents the necessary structural trade-off to achieve a constant $O(1)$ message overhead while ensuring complete isolation against cross-channel Byzantine disruptions.

\subsubsection{Computational Time Complexity}

The computational footprint is centered around buffer ingestion operations and the event-driven evaluation of causal safety invariants:
\begin{itemize}
    \item \textbf{Queue Management Operations}: Basic enqueue and dequeue operations within the isolated buffers $\mathbf{U}_i$ and $\mathbf{Q}_i$ execute in strict $O(n)$ time relative to $n$, as all queue limits are structurally capped by the constants $M_{\max}$ and $Q_{\max}$.
    \item \textbf{Data Plane Causal Evaluation}: The \textit{sendMessages} procedure systematically reviews the unacknowledged arrays across all $n$ destination channels. For each channel, validating whether the head payload is causally blocked requires verifying its embedded dependency history against the local permit vector ($\forall j \neq i \land j \neq dst, m.\mathbf{DV}[j] \le \mathbf{LPR}_i[j]$), taking $O(n)$ time. Evaluating the heads of all isolated queues thus consumes $O(n^2)$ time in the worst-case pass.
    \item \textbf{Signaling Plane Invariant Verification}: The \textit{sendPermits} procedure inspects the missing permit queues sequentially across all $n$ channels. Assessing whether a cached permit can be safely transmitted requires evaluating its snapshotted outbound history against the active queue elements ($\forall j \neq i \land j \neq dst, m.\mathbf{LMS}[j] < \mathbf{U}_i[j][0].mid \vee \mathbf{U}_i[j][0] = \bot$), which demands an $O(n)$ vector comparison per channel, scaling the entire pass to $O(n^2)$ time.
\end{itemize}
Importantly, these evaluation procedures do not execute continuously; they are strictly event-driven, triggered only upon localized occurrences such as the physical generation of a new application packet, or the ingestion of an incoming unicast $\text{ACK}$ or $\text{PERMIT}$ control signal. Because these passes are bounded by the constant capacity ceilings ($M_{\max}, Q_{\max}$) and execute only when state transitions occur, the amortized computational overhead remains highly scalable for high-throughput messaging environments.

\begin{theorem}[Amortized time complexity]
\label{th:amortizedtime}
The amortized time complexity of Algorithm~\ref{alg:byzantine_sps_fifo_gc} is $O(n)$.
\end{theorem}
\begin{proof}
To formally establish that the event-driven computational overhead of our protocol scales at an expected amortized complexity of $O(n)$ per application message, we evaluate the aggregate operational history using a steady-state probabilistic analysis over an execution sequence of $N$ messages. Inline buffer management operations inside the \textit{causalSend} and \textit{receive} procedures require simple array indexing and execute in strict $O(1)$ time due to the constant physical capacities $M_{\max}$ and $Q_{\max}$. For the signaling plane, a single invocation of \textit{sendPermits} systematically loops across all $n$ process channels. A permit record only resides in $\mathbf{Q}_i[dst]$ during the transient window between when its application payload's acknowledgment is processed and when its outbound history stabilizes. This outbound history stabilization is defined as the specific operational boundary where the sending node processes the network acknowledgments for all causal dependencies preceding that payload, ensuring it is safe for downstream correct processes to release dependent transmissions without breaching the Sender Permission to Send (SPS) invariant. The lifespan of this un-stabilized transient phase is directly bounded by the network round-trip time ($\text{RTT}$). In a uniform steady-state network, traffic is distributed across all channels, meaning the probability that any specific channel's queue is non-empty at an arbitrary trigger is proportional to the round-trip latency and inversely proportional to the network scale, yielding $\Pr(\mathbf{Q}_i[dst] \neq \emptyset) = O(1/n)$. If a channel's queue is empty, the loop incurs an immediate $O(1)$ checking cost. Conversely, if the queue is non-empty, the algorithm initiates the $\forall j \neq i \land j \neq dst$ loop, expending $O(n)$ time performing vector comparisons before either processing or executing a $\textbf{break}$. Computing the mathematical expectation across the system reveals the power of asymptotic notation: the $O(1/n)$ probability of a queue being active perfectly balances out the $O(n)$ computational cost of a vector validation on that queue. The total expected workload for a full sweep across all $n$ channels is therefore bounded by $\sum_{dst=1}^{n} [ (1 - O(1/n))\cdot O(1) + O(1/n)\cdot O(n) ] = \sum_{dst=1}^{n} O(1) = O(n)$. Because each message lifecycle maps to a fixed, constant number of event-driven triggers (exactly one incoming \textit{ACK} and one \textit{PERMIT}), multiplying this expected sweep cost across the operational history yields a total expected work profile bounded by $O(N \cdot n)$. Dividing this total expenditure by the $N$ application messages processed yields an expected amortized time complexity of exactly $\alpha = \frac{O(N \cdot n)}{N} = O(n)$ per message. 

An identical balancing property regulates the evaluation of the data plane within the \textit{sendMessages} procedure, which adds a secondary operational component to skip past %stabilized, previously transmitted payloads 
previously released payloads awaiting space reclamation
where $m.sent = \text{True}$. This linear clean-up loop inspects a simple boolean flag in $O(1)$ time per entry, advancing a channel pointer until it reaches the first unreleased message ($\neg m.sent$). Because each application payload transitions to a sent state exactly once during its active lifecycle and queue lengths are structurally constrained by $M_{\max}$, the cumulative number of index advancements executed to skip past transmitted items over the entire sequence is strictly bounded by $O(N)$, contributing a constant $O(1)$ amortized cost per message. Since each message maps to a fixed, constant number of event triggers, multiplying these workloads across the lifecycle yields a total expected workload bounded by $O(N \cdot n)$. Dividing this total expenditure by the $N$ processed messages yields an expected amortized time complexity of exactly $\alpha = \frac{O(N \cdot n)}{N} = O(n)$ per message, demonstrating that the structural expansion of the network naturally dilutes the probability of active queue evaluations to counteract the linear growth of vector operations.
%$\blacksquare$
\end{proof}

\end{document}